\documentclass[journal]{IEEEtran}

\usepackage{setspace,cite}

\usepackage[table]{xcolor}

\usepackage[tbtags]{amsmath}
\usepackage{amsbsy}
\usepackage{amssymb}
\usepackage{amsfonts}
\usepackage{amsthm}

\usepackage{graphicx}
\usepackage{algorithmic}
\usepackage{algorithm}
\usepackage{balance}
\usepackage{rotating}
\usepackage{multirow}
\usepackage{subcaption}
\usepackage{fancyvrb}
\usepackage{latexsym}
\usepackage{verbatim}

\newtheorem{theorem}{Theorem}
\newtheorem{proposition}{Proposition}
\newtheorem{lemma}{Lemma}

\newtheorem{definition}{Definition}

{\begin{list}               
    {$\bullet$ \hfill}{
        \setlength{\leftmargin}{\parindent}
        \setlength{\parsep}{0.04\baselineskip}
        \setlength{\itemsep}{0.5\parsep}
        \setlength{\labelwidth}{\leftmargin}
        \setlength{\labelsep}{0em}}
    }
{\end{list}}

\providecommand{\eref}[1]{\eqref{#1}}  
\providecommand{\cref}[1]{Chapter~\ref{#1}}

\providecommand{\fref}[1]{Figure~\ref{#1}}

\providecommand{\R}{\ensuremath{\mathbb{R}}}

\providecommand{\E}{\ensuremath{\mathbb{E}}}

\providecommand{\bydef}{\overset{\text{def}}{=}}

\renewcommand{\vec}[1]{\ensuremath{\boldsymbol{#1}}}

\providecommand{\calN}{\mathcal{N}}

\providecommand{\vh}{\vec{h}}

\providecommand{\vu}{\vec{u}}

\providecommand{\vw}{\vec{w}}
\providecommand{\vx}{\vec{x}}
\providecommand{\vy}{\vec{y}}
\providecommand{\vz}{\vec{z}}

\providecommand{\vnu}{\vec{\nu}}

\providecommand{\vyhat}{\boldsymbol{\widehat{y}}}

\providecommand{\Var}{\mathrm{Var}}

\newcommand{\argmin}[1]{\mathop{\underset{#1}{\mbox{argmin}}}}

\begin{document}

\title{Rethinking Atmospheric Turbulence Mitigation}

\author{Nicholas~Chimitt,~\IEEEmembership{Student~Member,~IEEE,}
        Zhiyuan~Mao,~\IEEEmembership{Student~Member,~IEEE,}
        Guanzhe~Hong,~\IEEEmembership{Student~Member,~IEEE,}
        and~Stanley~H.~Chan,~\IEEEmembership{Senior~Member,~IEEE}
\thanks{The authors are with the Department
of Electrical and Computer Engineering, Purdue University, West Lafayette,
IN, 47906 USA. Email: \{nchimitt, mao114, hong288, stanchan\}@purdue.edu.}}

\maketitle

\begin{abstract}
State-of-the-art atmospheric turbulence image restoration methods utilize standard image processing tools such as optical flow, lucky region and blind deconvolution to restore the images. While promising results have been reported over the past decade, many of the methods are agnostic to the physical model that generates the distortion. In this paper, we revisit the turbulence restoration problem by analyzing the reference frame generation and the blind deconvolution steps in a typical restoration pipeline. By leveraging tools in large deviation theory, we rigorously prove the minimum number of frames required to generate a reliable reference for both static and dynamic scenes. We discuss how a turbulence agnostic model can lead to potential flaws, and how to configure a simple spatial-temporal non-local weighted averaging method to generate references. For blind deconvolution, we present a new data-driven prior by analyzing the distributions of the point spread functions. We demonstrate how a simple prior can outperform state-of-the-art blind deconvolution methods. 
\end{abstract}


%
\IEEEpeerreviewmaketitle

\section{Introduction}
\subsection{Motivation and Contributions}
\IEEEPARstart{A}{tmospheric} turbulence is one of the most devastating distortions in long-range imaging systems. Under anisoplanatic conditions, a scene viewed through turbulence is perturbed by random warping and blurring that are spatially and temporally varying. Their magnitudes and directions are influenced by temperature, distance and viewing angle \cite{roggemann1996imaging}. Conventional turbulence restoration methods utilize standard image processing tools such as optical flow, lucky region fusion and blind deconvolution to recover images. While these methods are well-studied individually, they are agnostic to the physical  model governing the turbulence. For example, the warping due to turbulence is not an arbitrary non-rigid deformation but the result of a wave propagating through layers of random phase screens.

The goal of this paper is to revisit the turbulence restoration pipeline by asking a question: If we rigorously follow the Kolmogorov's model \cite{Kolmogorov1941}, how should each component in the turbulence restoration pipeline be configured so that the overall algorithm is grounded on physics. Our finding shows that when carefully designed, even very simple methods can perform better than sophisticated methods. See \fref{fig: exp house} for a comparison between different turbulence image restoration methods applied to a static scene.

To elaborate on this main statement, in this paper we investigate two steps of the turbulence restoration pipeline:

\begin{figure}[t]
	\centering
	\begin{tabular}{ccc}
		\includegraphics[width=0.3\linewidth]{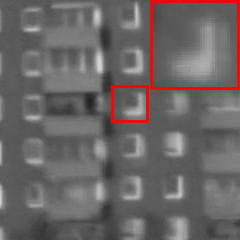}&
		\hspace{-2ex}\includegraphics[width=0.3\linewidth]{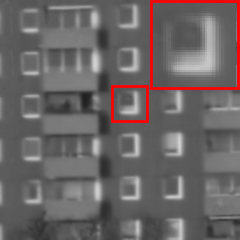}&
		\hspace{-2ex}\includegraphics[width=0.3\linewidth]{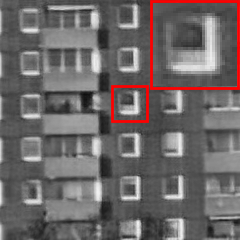}\\
		\small{(a) Input} & \hspace{-2ex} \small{(b) Lucky Region} & \hspace{-2ex} \small{(c) Anan. et al. \cite{Anantrasirichai2013}}\\
		\includegraphics[width=0.3\linewidth]{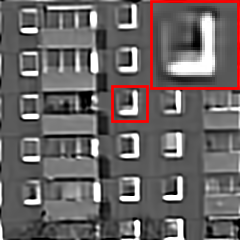}&
		\hspace{-2ex}\includegraphics[width=0.3\linewidth]{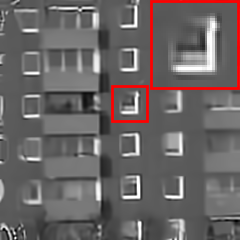}&
		\hspace{-2ex}\includegraphics[width=0.3\linewidth]{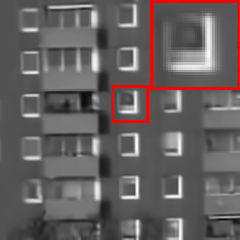}\\
		\small{(d) Lou et al. \cite{Lou2013}} & \hspace{-2ex} \small{(e) Zhu et al. \cite{Milanfar2013}} & \hspace{-2ex} \small{(f) Ours}
	\end{tabular}
	\caption{Comparison with state-of-the-arts on a real turbulence sequence \texttt{Building}. The proposed method is based on analyzing the physics of the turbulence.}
	\vspace{-2ex}
	\label{fig: exp house}
\end{figure}

\begin{figure*}[t]
	\centering
	\includegraphics[width=0.8\linewidth]{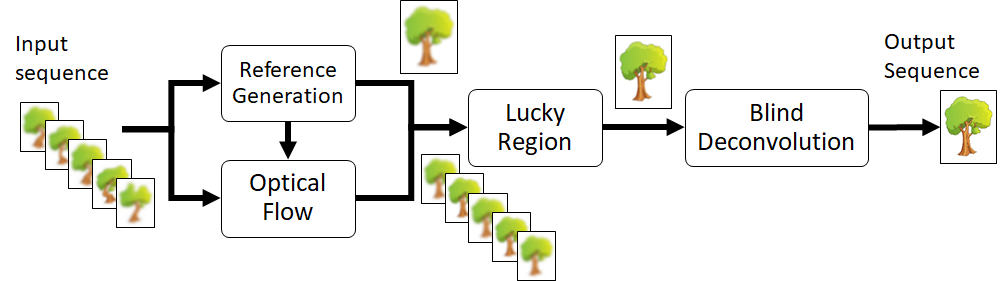}
	\caption{A typical turbulence restoration pipeline contains (i) reference frame generation, (ii) optical flow, (iii) lucky region fusion, and (iv) blind deconvolution. This paper focuses on the reference frame generation and the blind deconvolution steps.}
	\label{fig:general_static_block}
\end{figure*}

\begin{itemize}
	\item \textbf{Reference Frame}. Majority of the turbulence restoration pipelines involve optical flow and lucky region fusion. Both steps require a good reference, and typically the reference is computed from its neighboring frames. The number of frames plays a critical role here. If we use too few frames, then the turbulence pixels are not stabilized. However, if we use too many frames, then the image will be over-smoothed. Typically, the number of frames is unknown ahead of the restoration, and is tuned manually by the user. More sophisticated methods have built-in iterative mechanisms to update the reference while recovering the image, but these methods are time consuming. 
	
	\hspace{2ex} We study the reference frame generation problem from a physics point of view. Using a simplified Kolmogorov turbulence model, we assume that the turbulence point spread function is a kernel with random spatial offsets. By leveraging tools in high dimensional probability, in particular the large deviation theory, we rigorously analyze the number of adjacent frames required to produce a reasonable reference frame. Our theoretical results reveal potential flaws that could happen if we ignore the physics. (See Section 2.)
	\vspace{1ex}
	\item \textbf{Blind Deconvolution}. Blind deconvolution is used to remove the diffraction limited blur after the lucky region fusion step. Normally, at this stage one would assume that the blur is spatially invariant and so any off-the-shelf blind deconvolution method can be applied (e.g., deep neural network). However, rather than treating the blur as a completely unknown quantity, we argue that the diffraction limited blur in turbulence has a unique prior which can offer a good solution. We articulate the problem by building an accurate turbulence simulator (via wave propagation equations) to generate short exposure point spread functions, and learn the basis functions as well as the prior distribution. We show that a very simple Bayesian estimation is sufficient to provide high quality results. (See Section 3.)
\end{itemize}

\subsection{Related Work}
Turbulence image restoration is a well studied subject. The focus of this paper is the image processing approach. The underlying assumption is that the imaging system is passively acquiring images where the light source is incoherent. We do not use coherent light sources to illuminate the object and use adaptive optics to compensate for the phase shifts. Readers interested in active imaging approaches can consult, e.g., \cite{tyson2010principles}.

The image processing literature on turbulence is rich. In general, most of the methods follow a similar pattern: Reference frame, optical flow, lucky region, and deconvolution, as shown in  \fref{fig:general_static_block}. In the followings we briefly describe a few better known methods. 

\vspace{1ex}
\begin{itemize}
\item \textbf{Methods for Static Scenes}. Most of the turbulence image restoration methods in the literature are designed for static scenes, i.e., both the object and the background are not moving. Because the scene is static, all pixel movements are caused by turbulence. Therefore, one of the simplest approaches to generate a reliable reference frame is to take the temporal average. This has been used in many previous work, e.g., Lou et al. \cite{Lou2013}, Zhu and Milanfar \cite{Milanfar2013}, Gilles and Osher \cite{ Gilles2016}, and more recently Hardie et al. \cite{Hardie2017} and Lau et al. \cite{Lau2017}. 

\hspace{2ex} Once the reference frame is generated, it will be sent to an optical flow to estimate the motion. Depending on the complexity of the scene and the computing budget, optical flow can be as simple as the traditional block matching by Hardie et al. \cite{Hardie2017} or more customized methods such as B-spline by Zhu and Milanfar \cite{Milanfar2013}, or feature matching by Anantrasirichai et al. \cite{Anantrasirichai2013}.

\hspace{2ex} The output of the optical flow is a sequence of motion compensated frames. If the scene is static, these processed frames are aligned but the blur is spatially varying. The purpose of the lucky region fusion is to pick the sharp regions to form a so called ``lucky frame''. The way to determine the lucky region is very similar to the reference frame. Instead of using the temporal average, there is a term measuring the magnitude of the gradient \cite{Lou2013, Anantrasirichai2013}. Sharper frames typically have stronger gradients.

\hspace{2ex} The final step of the pipeline (\fref{fig:general_static_block}) is the blind deconvolution. In principle, the image sent to this stage should have been recovered except the diffraction limited blur. The goal of blind deconvolution is to remove the remaining blur. Since blind deconvolution is a generic problem, many methods can be used, e.g., Wiener filtering by Hardie et al. \cite{Hardie2017}  or minimizing energy functions such as total variation as in  \cite{Lou2013,Gilles2016,Milanfar2013}. 
\vspace{1ex}
\item \textbf{Methods for Dynamic Scenes}. Methods for dynamic scenes (moving foreground and a static background) have more variations. For example, instead of using the lucky region fusion, Gilles and Osher \cite{Gilles2016} proposed to use wavelet burst accumulation to boost high frequency components. For large moving objects, Nieuwenhuizen et al. \cite{Nieuwenhuizen2017} proposed to use a super resolution fusion step to ensure spatial consistency, and Huebner \cite{huebner2012} proposed a block matching algorithm and local image stacking. 

\hspace{2ex} Because of the moving foreground, one alternative approach is to use advanced segmentation algorithms to extract the foreground. The background can then be recovered using the static scene methods. Several papers are based on this idea, e.g., Oreifej et al. \cite{Oreifej2013}, Halder et al. \cite{Halder2015} and  Anantrasirichai et al. \cite{Anan2018}. However, a fundamental issue of segmentation-based methods is that in the presence of turbulence distortion, the object boundaries are very difficult to determine. Thus, artifacts are easily generated by these methods.

\vspace{1ex}
\item \textbf{Other Methods}. Beyond the above ``mainstream'' methods, there are also customized approaches for specific context, e.g., underwater imaging using deep learning \cite{zhengqin2018}, among a few others, e.g., bispectral analysis \cite{Wen2010}, infra-red \cite{Droege2012}, face recognition in turbulence \cite{Kamenetsky2018}, and holographic systems \cite{zhao2013}. 
\end{itemize}

In this paper, we focus on the pipeline shown in \fref{fig:general_static_block} because it is the most common pipeline which can be applied to both static and dynamic scenes. Among the components of the pipeline, we are particularly interested in the reference generation and the blind deconvolution step. The optical flow and the lucky region fusion are based on existing implementations. For example, for optical flow we use \cite{Liu2009optflow}, and for lucky region fusion we use a modified version of \cite{aubaillyLucky}.

\section{Reference Frame}
We first look at reference frame generation. The objective of this section is to present a simple and generic method. The idea is based on spatial-temporal non-local weighted averaging. After presenting the method, we will rigorously analyze the number of adjacent frames required for the averaging. We will show a few non-trivial results based on large deviation theory. 

To keep our notation simple, we consider only one dimension in space. Extension to two dimensions is straight forward.

\begin{figure}[t]
	\vspace{-2ex}
	\centering
	\includegraphics[width=\columnwidth]{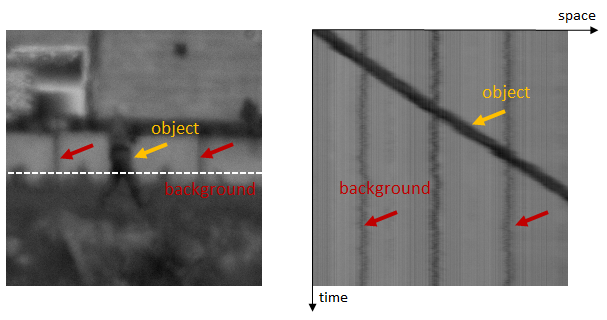}
	\vspace{-4ex}
	\caption{Space-time plot of a moving scene. The moving object will cause a significant displacement in the space-time plot. However, turbulence only perturbs a pixel by vibrating them around the center positions.}
	\label{fig:cross section}
\end{figure}

\subsection{Non-local Reference Generation Method}
The motivation of our reference generation method is illustrated in \fref{fig:cross section}, where we show a space-time plot of a typical turbulence distorted sequence. The moving object and a static pattern demonstrate very different trajectories: The moving object shows clear movement across the space, whereas a static pattern only vibrates at its center location. This difference suggests that if we pick a local patch distorted by turbulence, we should be able to find a similar patch in a small spatial-temporal neighborhood. In contrast, it will be more difficult to find a match for a motion patch.

Let $\vy_t \in \R^N$ be the $t$-th frame of the input video. Let $\vy_{i,t} \in \R^d$ be a $d$-dimensional patch located at pixel $i$. We set up a spatial-temporal search window of size $L \times T$, and then compute the distance between the current patch $\vy_{0,0} \in \mathbb{R}^d$ and all other patches $\{\vy_{i,t} \; |\;  i \in \Omega_L, \; t \in \Omega_T\}$ within the search window. This gives us
\begin{equation}
\Delta_{i,t} = \|\vy_{i,t}-\vy_{0,0}\|^2,
\end{equation}
where $i \in \Omega_L \bydef \{1,\ldots,L\}$ and $t \in \Omega_T \bydef \{1,\ldots,T\}$. Intuitively, we can think of $\Delta_{i,t}$ as a measure of how similar $\vy_{i,t}$ is to $\vy_{0,0}$. If $\vy_{0,0}$ is distorted by turbulence, then at least one of these $L$ patches in the adjacent frame should be similar to $\vy_{0,0}$. If $\vy_{0,0}$ is a moving object, then no patch in the window will be similar to $\vy_{0,0}$. See \fref{fig:test} for a pictorial illustration. Thus, for every frame, we can check the smallest residue among the spatial neighborhood, and define the temporal weight as 
\begin{equation}
w_t = \exp\left\{-\beta \min_{i \in \Omega_L} \left\{\Delta_{i,t}\right\} \right\}.
\label{eq: NL weight}
\end{equation}
Then, we use $w_t$ to compute the reference patch via
\begin{equation}
\vyhat_{0,0} = \frac{\sum_{t \in \Omega_T } w_t \vy_{0,t} }{ \sum_{t \in \Omega_T} w_t }.
\label{eq: NL ref patch construction}
\end{equation}
Note that $\vyhat_{0,0}$ overlaps when we move to another patch of the image. The overlapping can be taken care of by averaging out the overlapping pixels. 

We emphasize that \eref{eq: NL ref patch construction} is an extremely simple operation. It does not require object segmentation such as \cite{Oreifej2013, Halder2015, Anan2018}, and yet it is applicable to both static and dynamic scenes.

\begin{figure}
	\centering
	\begin{subfigure}{1\columnwidth}
		\centering
		\includegraphics[width=\columnwidth]{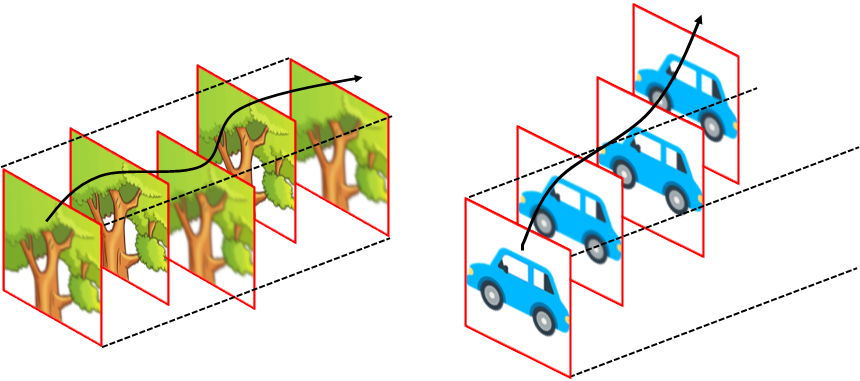}
		\label{fig:sub2}
	\end{subfigure}
	\vspace{-4ex}
	\caption{Given a current patch, the proposed reference frame generation method computes a weighted average across the adjacent frames. [Left] Turbulence: We average every patch along the time axis. [Right] Motion: We put small weights to patches that stays far from the center.}
	\label{fig:test}
\end{figure}

\subsection{Empirical Plot of $\beta$ for Static Scenes}
Like any other non-local averaging method, the hyper-parameter $\beta$ in \eqref{eq: NL weight} plays a critical role: If $\beta$ is too large, then we are dropping most of the adjacent frames, hence the result is reliant on $\vy_{0,0}$. If $\beta$ is too small, then we are being over-inclusive. In principle, $\beta$ should be chosen according to the turbulence. We now discuss how.

In Kolmogorov's model, turbulence is characterized by the refractive-index structure parameter $C_n^2$. $C_n^2$ is a function of the temperature, wavelengths and distance \cite{Hardie2017}. Integrating $C_n^2$ over the wave propagation path will give us the Fried parameter $r_0$. The reciprocal of $r_0$ normalized by the aperture diameter $D$ is a quantity $D/r_0$ we typically see in the literature. Larger $D/r_0$ means stronger turbulence \cite{roggemann1996imaging}. 

To initiate the discussion let us look at Figure~\ref{fig:beta_vs_ratio}. The figure shows an empirical plot of the best $\beta$ as a function of $D/r_0$ for a static point source. We generate this plot by simulating how the point source goes through the turbulence media for a specific $D/r_0$ ratio (see Appendix for details of our simulator). We then pick the largest $\beta$ \footnote{We pick the largest $\beta$ because it corresponds to the minimum number of frames.} that makes $\|\vz_0 - \widehat{\vy}_{0}\|^2 \le \epsilon$ for some tolerance level $\epsilon$, where $\vz_0$ is the ground truth and $\widehat{\vy}_0$ is the estimated reference point source. The experiment is repeated 10 times to average out the randomness of the individual turbulence. We report the mean and the standard deviation. 

Figure~\ref{fig:beta_vs_ratio} matches with our intuitions: As turbulence becomes stronger (larger $D/r_0$), we require more frames to average out the randomness (hence $\beta$ drops). But what is the exact relationship between the turbulence strength and the number of frames? In addition, why does $\beta$ stay at a constant when $D/r_0 < 1$?

\begin{figure}[t]
	\centering
	\includegraphics[width = \columnwidth]{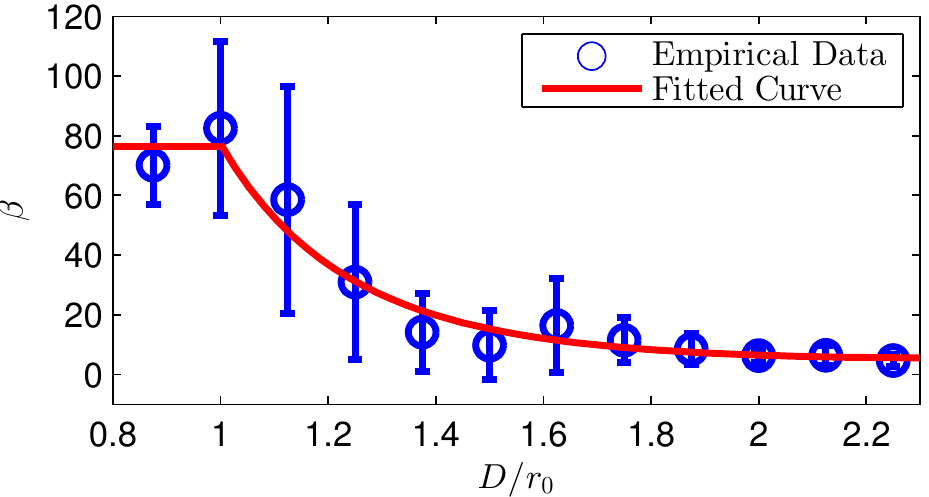}
	\caption{The optimal $\beta$ as a function of the turbulence strength $\frac{D}{r_0}$. The circles are the average of 10 independent trials and the error bar denotes one standard deviation.}
	\label{fig:beta_vs_ratio}
	\vspace{-2ex}
\end{figure}

\subsection{Short and Long Exposure PSFs}
To understand the behavior of $\beta$, we recall an old result by Fried \cite{Fried65} showing that 90\% of the disturbance due to turbulence is attributed to the random shifting of the points spread function (PSF). What this suggests is a simple model for the PSF by writing it as $h(x-\Theta)$, where $h(\cdot)$ is a fixed kernel and $\Theta$ is a random variable drawn from some distribution $p_{\Theta}(\theta)$. Therefore, given a fixed shape of the PSF $h(\cdot)$, we can shift it spatially to obtain an instantaneous PSF. The fluctuation of the shift is determined by  $p_{\Theta}(\theta)$. 

\begin{figure}[h]
    \centering
    \includegraphics[width=\linewidth]{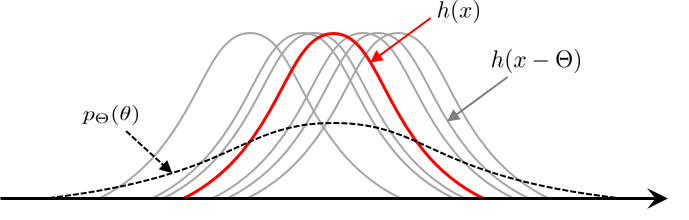}
    
    \vspace{-2ex}
    \caption{A simplified model of the PSF by shifting $h(x)$ using a random offset. The distribution of $\Theta$ is given by $p_{\Theta}(\theta)$.}
    \label{fig:kernel}
\end{figure}

In turbulence terminology, the un-perturbed PSF $h(\cdot)$ is called a \emph{short-exposure PSF} (short-PSF) \cite{roggemann1996imaging}. For mathematical analysis we assume that $h(\cdot)$ is a smoothing kernel with some parameter $\nu$.
\begin{definition}
	The point spread function (PSF) $h(\cdot)$ takes the form of
	\begin{equation}
	h_{\nu}(x) = \frac{1}{\nu}K\left(\frac{x}{\nu}\right)
	\end{equation}
	for some smoothing kernel $K$, and some constant $\nu>0$.
	\label{def:short-psf from smoothing kernel}
\end{definition}
In this definition, the number $\nu$ controls the ``bandwidth'' of $h_{\nu}$ without changing its ``volume''. The expectation $\E_{\Theta}[h_{\nu}(x-\Theta)]$ is called the \emph{long-exposure PSF} (long-PSF). The long-PSF can be shown as the convolution of the short-PSF and the distribution $p_{\Theta}(\theta)$:
\begin{equation}
\begin{split}
\E_{\Theta}[h_{\nu}(x-\Theta)] 
&= \int h_{\nu}(x-\theta)p(\theta)d\theta \\
&= (h_{\nu} \circledast p_{\Theta})(x).
\end{split}
\end{equation}
To analyze the practical situation, we also define a finite sample estimate
\begin{equation}
\label{long_exposure}
\widetilde{h}_{\nu}(x) \bydef \frac{1}{T} \sum_{t=1}^{T} h_{\nu}(x-\theta_t),
\end{equation}
where $\{\theta_t\}$ are i.i.d. copies of $\Theta$. Our goal is to analyze how $T$ changes with $\Var[\Theta]$, as $T$ is a proxy for $\beta$ and $\Var[\Theta]$ is a proxy for $D/r_0$.

\vspace{1ex}
\noindent \textbf{Remark}: What is the distribution $p_{\Theta}(\theta)$? If we look at how the short exposure and the long exposure PSFs are derived in the literature, we can see that the distribution is in fact a Gaussian. See Roggemann \cite{roggemann1996imaging} (Section 3.5). In particular, take the ratio of the long exposure PSF in Equation 3.125 and the short exposure PSF Equation 3.135. Because of this, we model the distribution $p_{\Theta}(\theta)$ as a Gaussian with zero mean and variance $\sigma^2$: $\Theta \sim \calN(0,\sigma^2)$.

\subsection{Concentration of Measure Results}
We now present the main theoretical result. The following theorem shows that, as $T$ increases, the finite sample estimate $\widetilde{h}_{\nu}(x)$ will approach its expectation $(h_{\nu} \circledast p_{\Theta})(x)$ with high probability. 
\begin{theorem}(Concentration of $\widetilde{h}_{\nu}(x)$). For any $\epsilon>0$,
	\begin{equation}
	\begin{aligned}
	& \sup_{x\in\mathbb{R}} \; \mathbb{P}\left(\left| \widetilde{h}_{\nu}(x) - (h_{\nu} \circledast p_{\Theta})(x) \right| > \epsilon  \right) \le \\ 
	&\quad\quad  2\exp\left\{-\frac{\epsilon^2 T}{2\sup\limits_{x\in\mathbb{R}} V_{\nu}(x,\sigma)+2\nu^{-1}M\epsilon/3}\right\} \bydef \alpha,
	\end{aligned}
	\label{semi-uniform concentration inequality}
	\end{equation}
	where $M$ is an upper bound of $K(\cdot)$, i.e., $0 \le K(x) \le M$ for all $x \in \R$, and $V_{\nu}(x,\sigma) \bydef \Var[h_{\nu}(x-\Theta)]$.
\end{theorem}
\begin{proof}
See Appendix for proof. 
\end{proof}

There are several important implications of the theorem: 
\begin{itemize}
	\item Fix $\epsilon$. As number of frames $T$ increases, the probability of getting a large deviation is exponentially decaying. In terms of turbulence, it says that the finite sample PSF converges to the long-PSF, something we expect and something well-known.
	\item While $T$ can be arbitrarily large, in practice we always use the smallest $T$ such that the probability meets the tolerance upper bound $\alpha$. This is coherent with how $\beta$ is generated in \fref{fig:beta_vs_ratio}. 
	\item The smallest $T$ is determined by $V_{\nu}(x,\sigma)$. As the turbulence becomes stronger, $V_{\nu}(x,\sigma)$ increases. The theorem then predicts that we need a large $T$ to achieve a tolerance $\alpha$. This is precisely what is happening in \fref{fig:beta_vs_ratio} for $D/r_0 \ge 1$: the stronger turbulence we have, the more frames we need.
	
	\item A big surprise comes when we compute the variance:
	\begin{equation*}
	\begin{aligned}
	V_{\nu}(x,\sigma) = & \Var_{\Theta}[h_{\nu}(x-\Theta)] \\ 
	= & \E_{\Theta}[h_{\nu}^2(x-\Theta)] - \E_{\Theta}[h_{\nu}(x-\Theta)]^2 \\
	= & (h_{\nu}^2 \circledast p_{\Theta})(x) - (h_{\nu} \circledast p_{\Theta})^2(x).
	\end{aligned}
	\end{equation*}
	If we plot $\sup_{x\in\mathbb{R}} V_{\nu}(x,\sigma)$ as a function of $\sigma$ (i.e., $\sqrt{\Var[\Theta]}$), we obtain \fref{fig:supremum variance vs sigma}. As the turbulence strength $\sigma$ increases,  $\sup_{x\in\mathbb{R}}V_{\nu}(x,\sigma)$ rises and then \emph{drops}! In other words, the theorem predicts that when turbulence is extremely strong, we actually need very few frames. This is counter intuitive.
	\item More problematically, the existence of a maximum of $\sup_{x\in\mathbb{R}}V_{\nu}(x,\sigma)$ is guaranteed for \emph{any} PSF $h_{\nu}(\cdot)$. See Theorem 2 below. Therefore, the problem is universal.
\end{itemize}

\begin{theorem}
	If $h_{\nu}$ is Lipschitz continuous on $\mathbb{R}$, or continuous but with compact support on $\mathbb{R}$, then $\sup_{x\in\mathbb{R}}V_{\nu}(x,\sigma)$, as a function of $\sigma$, attains a global maximum on $(0,\infty)$.
	\label{maximum attained for lipschitz and continuous}
\end{theorem}
\begin{proof}
The proof requires setting up several preliminary results in real analysis. We outline a sketch of the proof in the Appendix. Readers interested in the complete proof can check the Supplementary material.
\end{proof}

\begin{figure}[t]
	\centering
	\includegraphics[width=\columnwidth]{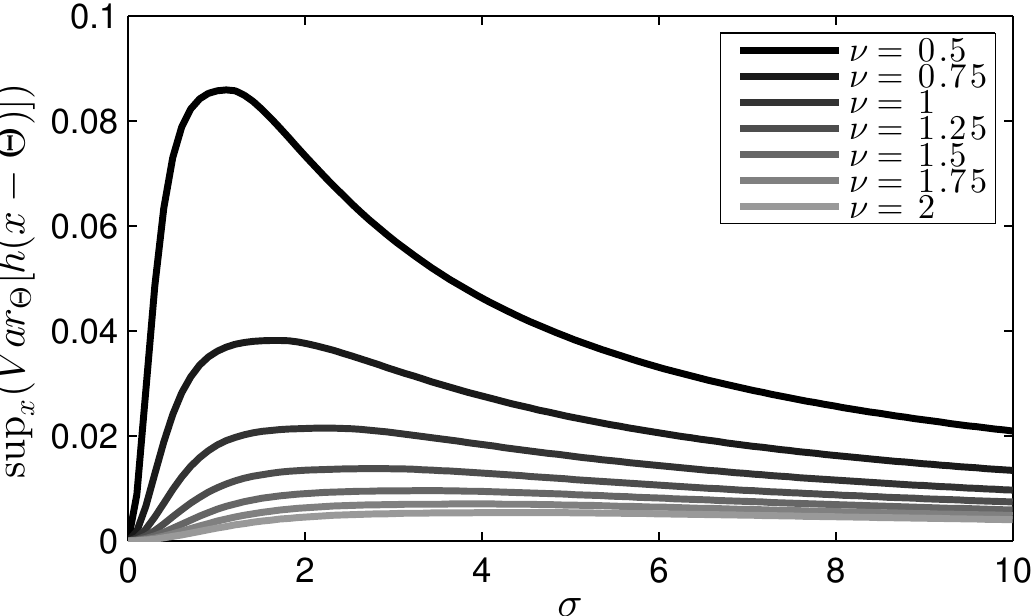}
	\vspace{-4ex}
	\caption{$\sup_{x\in\mathbb{R}}V_{\nu}(x,\sigma)$ vs. $\sigma$. The smoothing kernel $K$ is taken to be the standard Gaussian function here. The existence of a peak suggests that the blur and shift of a turbulence have to be increase simultaneously, or otherwise the result will be physically invalid.}
	\label{fig:supremum variance vs sigma}
	\vspace{-2ex}
\end{figure}

So what causes the disparity between our intuition and the theorem? The theorem is perfectly fine. What is not correct is a false assumption we overlooked. In the real turbulence setting, $\Var[\Theta]$ (i.e., $\sigma$) cannot be arbitrarily large. More specifically, the \emph{allowable} shift $\sigma$ for any real turbulence must be bounded by the bandwidth $\nu$ of the PSF. When $\sigma$ is bounded, we can show that $\sup_{x\in\mathbb{R}}V_{\nu}(x,\sigma)$ only operates in the increasing regime before reaching the maximum. Proposition \ref{prop: boxcar kernel sigma over nu} below shows a special case where $K$ is a boxcar kernel (so that we can derive analytic solution.)

\begin{proposition}
	If $K$ is the boxcar kernel, then for $\sigma\le \left(\Phi^{-1}\left(\frac{3}{4}\right)\right)^{-1} \nu \approx 1.48\nu$, $V_{\nu}(x=0,\sigma)$ is an increasing function in $\sigma$, where $\Phi(\cdot)$ is the normal CDF. 
	\label{prop: boxcar kernel sigma over nu}
\end{proposition}
\begin{proof}
See the Appendix for proof. 
\end{proof}

The implication of Proposition~\ref{prop: boxcar kernel sigma over nu} is significant. It suggests that many warping models based on the ad-hoc non-rigid deformations could be flawed if not modeled properly. A realistic turbulence must have shift and blur happening at the same time. A crude approximation of the shift with respect to the blur is $\sigma \approx 1.48\nu$. An empirical plot of this result is shown in Figure \ref{fig: sigma vs nu}. 

\begin{figure}[h]
	\centering
	\includegraphics[width=0.9\columnwidth]{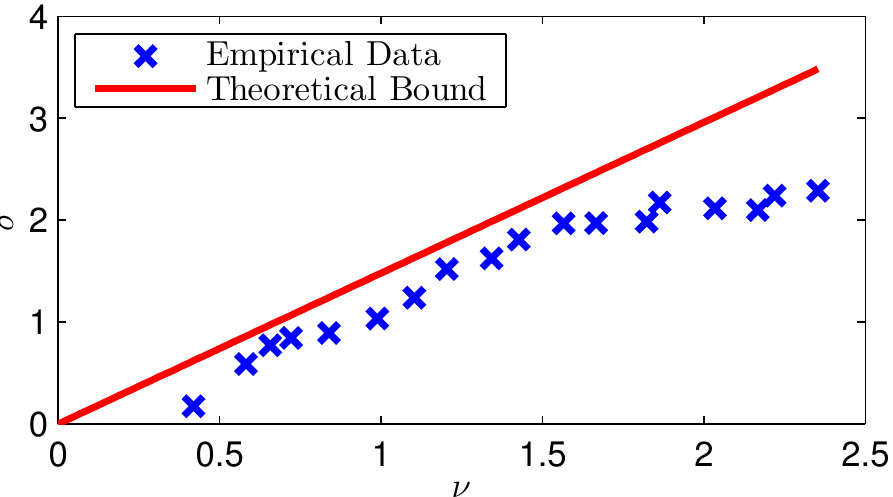}
	\caption{Verification of Proposition 1. The empirical data are generated by first simulating a PSF and measuring the bandwidth $\nu$ and the shift $\sigma$. The empirical data matches with the theory which predicts that $\sigma \le 1.48\nu$.}
	\label{fig: sigma vs nu}
\end{figure}

\noindent \textbf{Remark}. The results in Proposition~\ref{prop: boxcar kernel sigma over nu} can be improved by assuming a Gaussian kernel instead of the boxcar kernel. However, to do so we need to use numerical methods because the Gaussian kernel does not allow closed-form analysis.

\subsection{The Weak Turbulence Case}
We now consider $\beta$ when $D/r_0<1$. In this regime, the random shifting caused by the turbulence is insignificant compared to the diffraction limit of the PSF. In the turbulence literature, this phenomenon is known as that the ``seeing error'' due to the turbulence effects being overridden by the Airy disc \cite{tyson2010principles}. Putting in our terminology, we can say that $\Var[\Theta] \approx 0$ and so $h_{\nu}(x-\Theta)\approx h_{\nu}(x)$.

When $h_{\nu}(x-\Theta)\approx h_{\nu}(x)$, the random shifting is negligible. Thus, we can use \emph{any} number of frames, including just one frame or many frames. The numerical result in Figure \ref{fig:beta_vs_ratio} shows that $\beta$ is at the maximum (i.e., as few frames as possible.) However, a different choice of $\beta$ will perform equally well.

\subsection{Moving Objects}
In the presence of motion, the short-PSF changes from $h_{\nu}(x-\Theta)$ to $h_{\nu}(x(t) - \Theta)$, where $x$ is now a function of time. Assuming constant velocity so that $x'(t) = c$ for some $c \in \R$, we can show that
\begin{align}
&\E\left[\frac{1}{T}\sum_{t=1}^T h(x+ct-\Theta) \right] \notag \\
&\approx \E\left[\frac{1}{T}\sum_{t=1}^T \Big\{h(x -\Theta) + h'(x - \Theta) ct\Big\}  \right] \notag \\
&= 
\underset{\mbox{turbulence}}{\underbrace{(h_{\nu} \circledast p_{\Theta})(x)}} + 
\underset{\mbox{motion}}{\underbrace{\frac{c(T+1)}{2}(h_{\nu}' \circledast p_{\Theta})(x)}}.
\end{align}
Therefore, the perturbation caused by motion is captured by the bias $\frac{c(T+1)}{2}(h_{\nu}' \circledast p_{\Theta})(x)$. If we assume that $|(h_{\nu}' \circledast p_{\Theta})(x)| \le R$ for all $x$, then 
\begin{equation}
\begin{split}
\frac{c(T+1)}{2}(h_{\nu}' \circledast p_{\Theta})(x) &\le B\\
\;\;\;\Longleftrightarrow\;\;\;
T &\le \frac{2B}{cR}-1.
\end{split}
\label{eq: motion bound}
\end{equation}
The right hand side of \eref{eq: motion bound} provides an upper bound on $T$. Let us differentiate the $T$'s for static and dynamic case. We denote the number of frames for the static case as $T_{\mathrm{s}}$, and that for dynamic case as $T_{\mathrm{d}}$. The overall $T$ is the smaller of the two: $T = \min\{T_{\mathrm{s}}, T_{\mathrm{d}}\}$.

Now we can comment on the implication of \eref{eq: motion bound}. When $c = 0$, i.e., static scene, we have $T_{\mathrm{d}} < \infty$. In this case, the actual number of frames $T$ is determined by the turbulence $T_{\mathrm{s}}$. When the velocity $c$ grows, $T_{\mathrm{d}}$ drops. If $T_{\mathrm{d}}$ drops below $T_{\mathrm{s}}$, then $T = T_{\mathrm{d}}$. Thus for very fast moving objects, the number of frames is limited by $T_{\mathrm{d}}$. 

In terms of our algorithm, the factor $\frac{2B}{R}$ serves the role of the spatial search window size $L$. If $L$ is small, then even patches with small motion will be skipped. This shows the two distinctive roles of $\beta$ and $L$. $\beta$ is used to measure the turbulence, whereas $L$ is used to measure the object velocity. The typical range of $\beta$ is given by \fref{fig:beta_vs_ratio}, and the typical value of $L$ is $11 \times 11$ for a $500 \times 500$ image.

\section{Blind Deconvolution}
In this section we look at the blind deconvolution step. Our basic argument is that the blind deconvolution for a turbulence problem does not need to be very complicated because the turbulence PSF is well structured. Our goal is to exploit this structure and to propose a simple but effective blind deconvolution method. 

\subsection{Blind Deconvolution Algorithm}
Recall that the blind deconvolution is applied to the output of the lucky region fusion step. We need a blind deconvolution because we do not know the blur and the latent image. 

The proposed algorithm begins with a standard alternating minimization:
\begin{align}
\vz^{k+1} &= \argmin{\vz}\; \|\vy - \vh^k \circledast \vz\|^2 + \lambda g(\vz) \label{eq: vx subproblem}\\
\vh^{k+1} &= \argmin{\vh}\; \|\vy - \vh \circledast \vz^{k+1} \|^2 + \gamma r(\vh), \label{eq: vv subproblem}
\end{align}
where $\vy$ is the output of the lucky region fusion, $\vh$ is the unknown PSF, and $\vz$ is the latent clean image. The equation for $\vz^{k+1}$ is to update the latent image $\vz$ by using the currently estimated PSF $\vh^{k}$. Similarly, the equation for $\vh^{k+1}$ is to update the PSF by using the currently estimated $\vz^{k+1}$. In these two equations, $g(\cdot)$ and $r(\cdot)$ are regularization functions for $\vz$ and $\vh$, respectively. For performance, $g(\cdot)$ is chosen as the Plug-and-Play prior \cite{chan2017plug} using BM3D as the denoiser. 

At the output of the lucky region fusion step, only the sharpest frames are aggregated to form a diffraction limited image \cite{Milanfar2013}. Thus, the PSF for the blind deconvolution step is a short-PSF plus minor distortions in phase and magnitude (due to uncertainty caused by finite sample averaging and optical flow). To encapsulate the features of these distorted short-PSFs, we adopt a simple linear model by writing $\vh$ as a set of basis vectors $\vh = \sum_{i=1}^m w_i \vu_i$, where $\{\vu_i\}_{i=1}^m$ are to be trained. By incorporating this into the algorithm, we replace \eref{eq: vv subproblem} by two steps:
\begin{align}
\vw^{k+1} &= \argmin{\vw} \; \left\|\vy - \left(\sum_{i=1}^m w_i \vu_i\right) \circledast \vz^{k} \right\|^2 + \gamma r(\vw)  \label{eq: vw subproblem 1} \\
\vh^{k+1} &= \sum_{i=1}^p w_i^{k+1} \vu_i,
\label{eq: vw subproblem 2}
\end{align}
where $r(\vw)$ is a regularization function on $\vw$. The overall blind deconvolution algorithm now consists of three steps: updating the image estimate \eref{eq: vx subproblem}, followed by estimating the weight \eref{eq: vw subproblem 1} and constructing the PSF estimate \eref{eq: vw subproblem 2}. The algorithm repeats until convergence. \footnote{Beyond these major steps, we adopt two standard practice. When estimating the PSF in \eref{eq: vw subproblem 1}, we replace $\vy$ and $\vx$ by their gradients $\nabla \vy$ and $\nabla \vx$ as suggested by \cite{Shan2018}. For large images, we use a coarse-to-fine propagation by first estimating the PSF at coarse scale, and then progressively improve the resolution \cite{Xu_2018}.} Like the reference generation step, we emphasize that the proposed blind deconvolution method is extremely simple. However, we will show that by carefully choosing $\{\vu_i\}_{i=1}^m$ and $r(\vw)$, this method is sufficient to produce good results.

\subsection{Basis Functions}
Given that we are working on turbulence, one straight-forward approach to construct the basis functions is to simulate a large number of training PSFs and learn the principal components. To this end, we simulate 40,000 short-PSFs, each of size $15 \times 15$. These 40K short-PSFs cover a wide range of $C_n^2$ from $5\times 10^{-17}$ to $5\times 10^{-16}$, which is sufficient to model mild to medium turbulence. After generating these training samples we learn the principal components using the standard PCA. 

\fref{fig:psf_bases} shows a few snapshots of the generated short-PSFs at $C_n^2 = 10^{-17}$. We see that in general the short-PSFs are highly structured. All PSF shown have a similar mean and small distortions around the center. If we look at the basis functions, we see that the basis functions are nothing but a set of directional filters. These directional filters are orthogonal.

\begin{figure}[h]
	\centering
	\begin{tabular}{c}
		\includegraphics[width=0.95\linewidth]{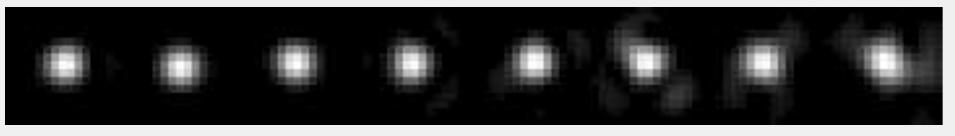}\\
		\includegraphics[width=0.95\linewidth]{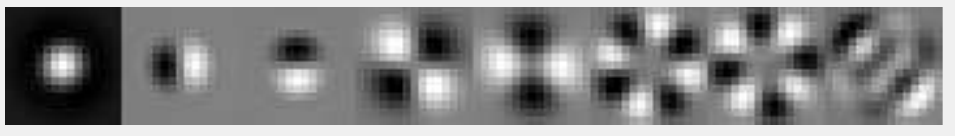}
	\end{tabular}
	
	\vspace{-2ex}
	\caption{[Top] Simulated short-PSFs, displayed according to increasing turbulence level from left to right. [Bottom] The first 8 principle components.}
	\label{fig:psf_bases}
	\vspace{-2ex}
\end{figure}

\subsection{Prior Distribution of the Weights}
Once the bases $\{\vu_1,\ldots,\vu_m\}$ are defined, we can examine the weight $\vw = [w_1,\ldots,w_m]^T$. To this end, we conduct an experiment to approximate a simulated short-PSF using as few bases as possible. This leads to an $\ell_0$ optimization by minimizing the number of active bases while bounding the error by $\tau$:
\begin{equation}
\widehat{\vw} = \argmin{\vw}\;\; \|\vw\|_0 \;\; \mbox{s.t.} \;\; \left\|\vh_{\mathrm{sim}} - \sum_{i=1}^m w_i \vu_i \right\|^2 \le \tau.
\label{eq: w optimizaiton}
\end{equation}
Here, $\vh_{\mathrm{sim}}$ denotes a simulated training short-PSF. We repeat the experiment for 100,000 different $\vh_{\mathrm{sim}}$, and we plot the histogram of each weight $\widehat{w}_i$ over these 100,000 trials. The empirical histogram is shown in \fref{fig:prior distribution}. 

\begin{figure}[h]
	\centering
	\includegraphics[width=\linewidth]{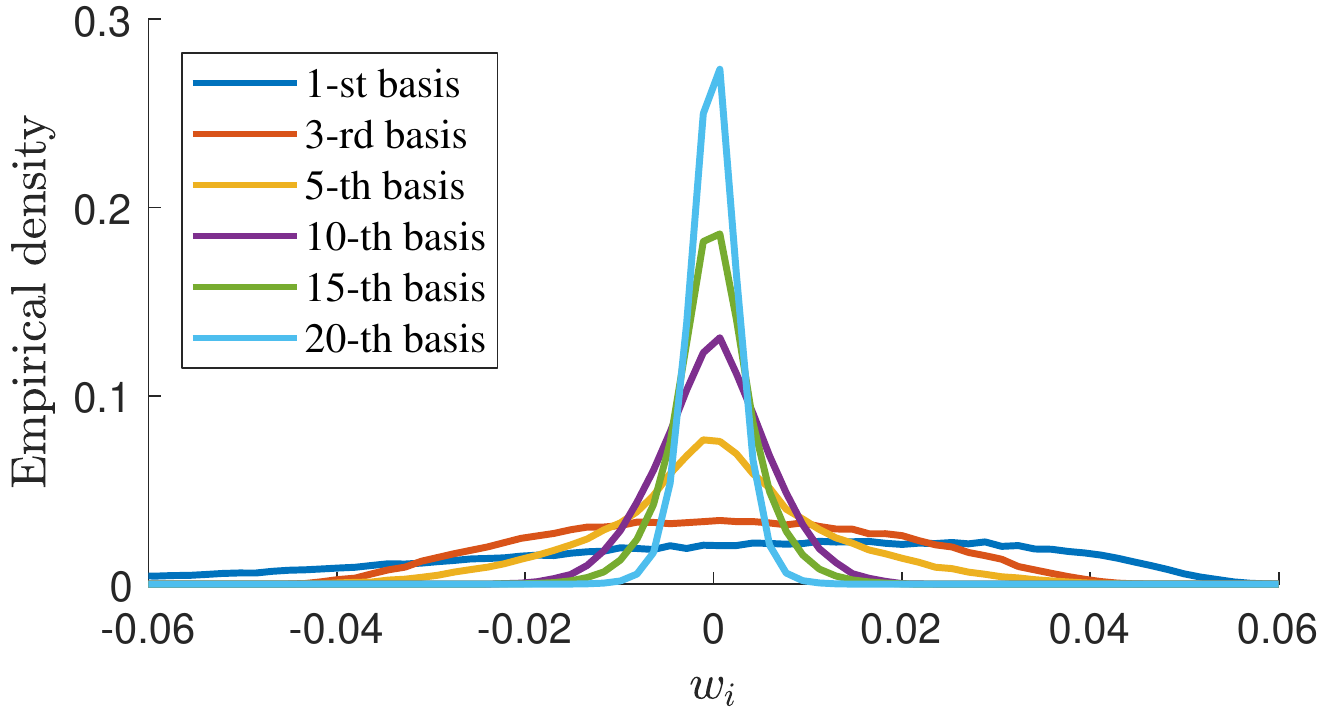}
	\caption{Empirical distributions of the weights by solving \eref{eq: w optimizaiton} with 100,000 different simulated PSF $\vh_{\mathrm{sim}}$. The empirical distribution demonstrates a double-sided exponential distribution.}
	\label{fig:prior distribution}
\end{figure}

By inspecting the histograms in \fref{fig:prior distribution}, we notice a double-sided exponential distribution. This suggests a product of exponential distributions for $\vw$:
\begin{equation}
p(\vw) = \exp\left\{-\sum_{i=1}^m \frac{|w_i|}{d_i}\right\},
\end{equation}
where $d_i$ measures the standard deviation of the individual exponential. Consequently we can define the regularization by taking the negative log: $r(\vw) = \sum_{i=1}^m \frac{|w_i|}{d_i}$.

\section{Experimental Results}

\subsection{Ablation Study of Reference Generation}
We first consider an ablation study of the reference generation method. To allow quantitative comparison, we simulate a 100-frame static-scene turbulence-distorted sequence at several different $C_n^2$'s. We use two metrics in this experiment. The first metric is the PSNR between the generated reference and the ideal short-exposure image. The ideal short-exposure image is generated by filtering the image with a short-PSF. This short-PSF is obtained by centroiding and averaging the simulated PSFs. The second metric is the PSNR between the final restoration result and the ground truth. That is, we fix the components of the pipeline except the reference generation step. The goal is to test the influence of the reference image. 

The results of this experiment are shown in Table~\ref{tab:reference 1}.  The first half of the table shows that the proposed reference generation method produces a reference that is closer to the ideal short-exposure image than the conventional temporal averaging. \fref{fig:ref_frame_comp} shows a visual comparison. The second half of the table shows that the influence of the reference frame is significant especially for large $C_n^2$. One reason is that as $C_n^2$ grows, the turbulence distortion becomes stronger and so the reference is over-smoothed. Feeding this over-smoothed reference to optical flow and lucky frame will degrade the performance significantly.

\begin{table}[h]
	\centering
	\begin{tabular}{ccccc}
		\hline
		\hline
		$C_n^2$ ($\times 10^{-17}$) & $5$ & $8$ & $20$ & $50$  \\
		\hline
		\multicolumn{5}{c}{(PSNR between reference and ideal short exposure)}\\
		Tmp Avg & 39.82 & 38.67 & 33.48 & 28.57 \\
		Ours    & \textbf{40.23} & \textbf{39.60} & \textbf{35.85} & \textbf{31.26} \\
		\hline
		\multicolumn{5}{c}{(Overall PSNR by changing reference in the pipeline)}\\
		Tmp Avg & 27.72 & 27.47 & 25.79 & 23.69 \\
		Ours    & \textbf{27.77} & \textbf{27.67} & \textbf{27.29} & \textbf{26.47} \\ 
		\hline
	\end{tabular}
	\caption{[Top] Comparing reference with respect to the ideal short exposure image. [Bottom] Ablation study by changing the reference in the pipeline.}
	\label{tab:reference 1}
\end{table}

In addition to the synthetic experiment, we also test on real moving sequences shown in \fref{fig:ref_frame_comp 2}. As the person in the sequence moves, temporal averaging will blur out the person. In contrast, the proposed method can retain the person while stabilizing the background.

\begin{figure}[h]
	\centering
	\begin{tabular}{ccc}
		\includegraphics[height=15ex]{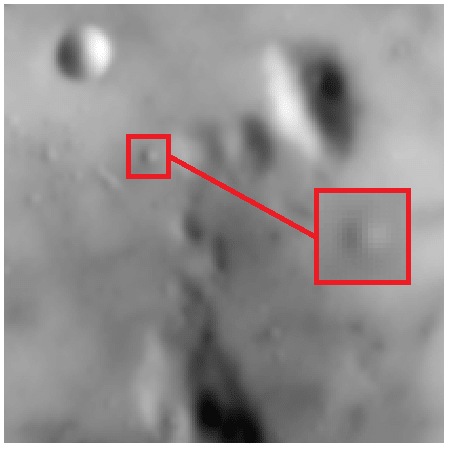} & 
		\hspace{-2ex} 
		\includegraphics[height=15ex]{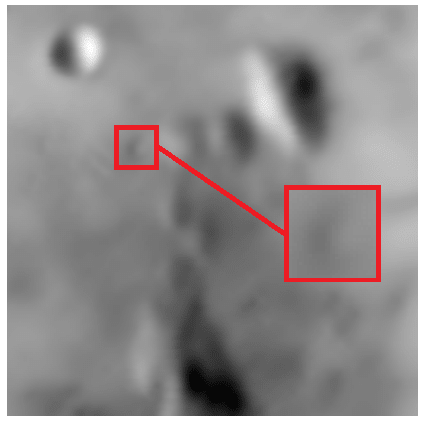} & 
		\hspace{-2ex}
		\includegraphics[height=15ex]{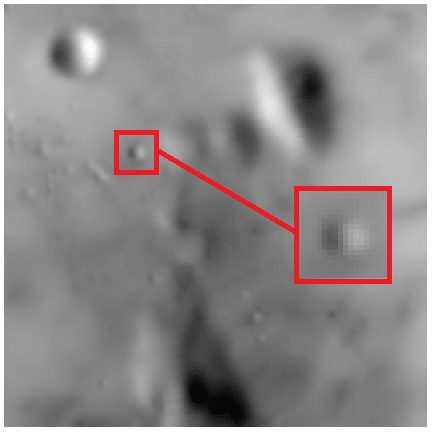}\\
		\small{Ideal Short Exp.} & \small{Tmp Avg} & \hspace{-2ex} 
		\hspace{-2ex} \small{Ours}
	\end{tabular}
	\vspace{-1ex}
	\caption{Synthetic experiment by simulating turbulence distorted images. This figure compares the generated reference with respect to the ideal short exposure image. }
	\label{fig:ref_frame_comp}
\end{figure}

\begin{figure}[h]
	\centering
	\begin{tabular}{ccc}
		\hspace{-2ex}
		\includegraphics[trim={3cm 0 2cm 0},clip, height=15ex]{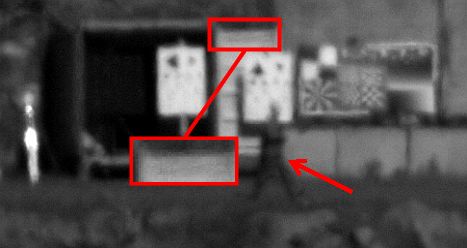} &
		\hspace{-2ex}
		\includegraphics[trim={3cm 0 2cm 0},clip, height=15ex]{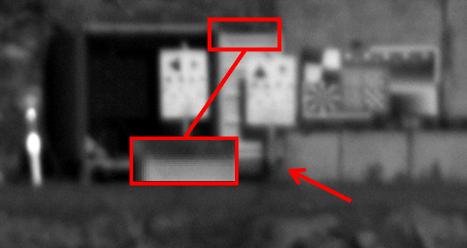} &
		\hspace{-2ex}
		\includegraphics[trim={3cm 0 2cm 0},clip, height=15ex]{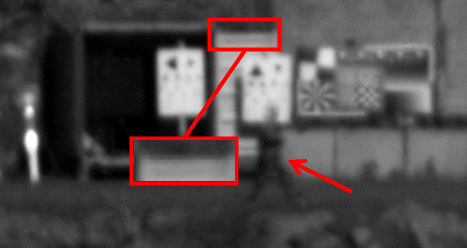} \\
		\hspace{-2ex} \small{(a) Raw input} &
		\hspace{-2ex} \small{(b) Tmp Avg} &
		\hspace{-2ex} \small{(c) Ours}
	\end{tabular}
	\caption{[Left] Raw input. [Middle] Temporal averaging: The man is blurred over 100 frames. [Right] Ours: The man remains in the image while the background is stabilized.}
	\label{fig:ref_frame_comp 2}
	\vspace{-2ex}
\end{figure}

\subsection{Ablation Study of Blind Deconvolution}
The second experiment is an ablation study to test the effectiveness of the blind deconvolution algorithm. The competing methods we consider include a classical  method by Shan et al. \cite{Shan2018}, and two very recent deep neural networks by Chakrabarti \cite{Chakrabarti2016} and Xu et al. \cite{Xu_2018}. We downloaded the original implementations of these methods and used the pre-trained models. Internal parameters (for \cite{Shan2018}) are fined tuned to maximize the performance. For our proposed method, we fix $\lambda = 0.05$ and $\gamma = 1\times10^{-4}$ for all experiments reported in this paper.

We use the 24 images in the Kodak image dataset for experiment. Every image is blurred with 50 random short-PSFs under 5 different turbulence levels with $C_n^2$ from $5\times 10^{-17}$ to $5\times 10^{-16}$. Thus each method at every turbulence level consists of 1200 testing scenarios. Since this is a simulated experiment, we have access to the ground truths to compute the PSNR values. The average PSNR over the 50 random PSFs are shown in Table \ref{tab:PSNR}.

For visual quality comparison, we demonstrate a result with the USAF resolution chart using turbulence at a level of $C_n^2 = 2.5 \times 10^{-16}$. It can be seen in \fref{fig:USAFresult} that the result using our method contains the least amount of artifacts. The estimated PSF is also more structured and interpretable than the deep neural networks \cite{Chakrabarti2016, Xu_2018}.

\begin{table}[h]
	\centering
	\begin{tabular}{cccccc}
		\hline
		\hline
		$C_n^2$ ($\times 10^{-17}$) & $5$ & $15$ & $25$ & $35$ & $45$ \\
		\hline
		Shan et al. \cite{Shan2018} & 27.25 & 26.52 & 26.59 & 25.91 & 23.89 \\
		Chakrabarti \cite{Chakrabarti2016} & 27.32 & 26.69 & 26.31 & 24.80 & 24.49 \\         
		Xu et al. \cite{Xu_2018} & 26.68 & 26.08 & 26.04 & 24.45 & 24.21 \\         
		Ours & \textbf{27.69} & \textbf{27.02} & \textbf{27.00} & \textbf{26.24} & \textbf{24.58}\\
		\hline
	\end{tabular}
	\caption{Ablation study of blind deconvolution algorithms. The PSNRs are averaged over 24 images of the Kodak image dataset, and every image is tested 50 times for different random short-PSFs. Thus every data point reported in this table is an average over 1200 testing scenarios.}
	\label{tab:PSNR}
\end{table}

\begin{figure}[h]
	\centering
	\begin{tabular}{ccc}
		\includegraphics[width=0.3\linewidth]{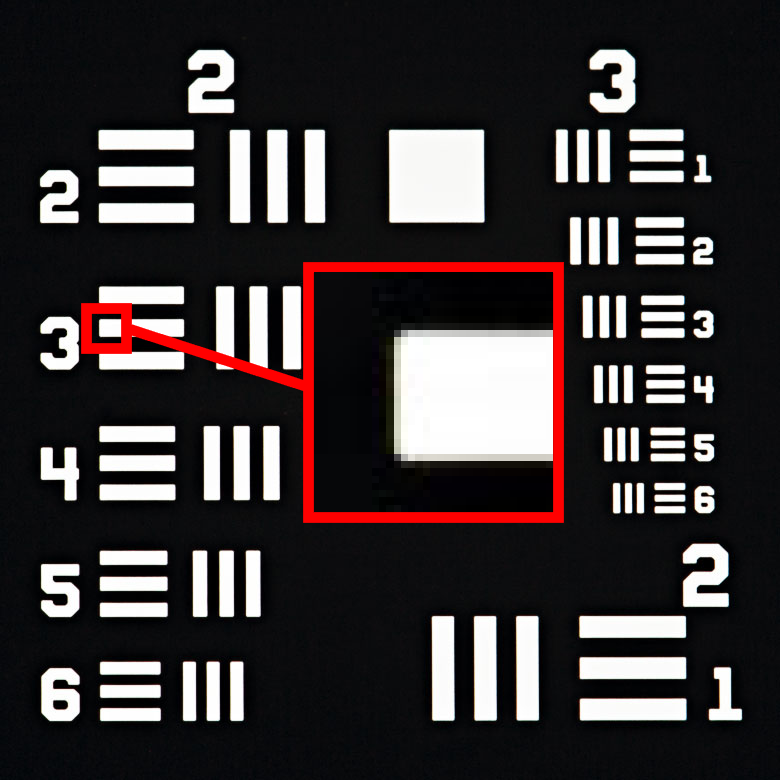}   & \hspace{-2ex} \includegraphics[width=0.3\linewidth]{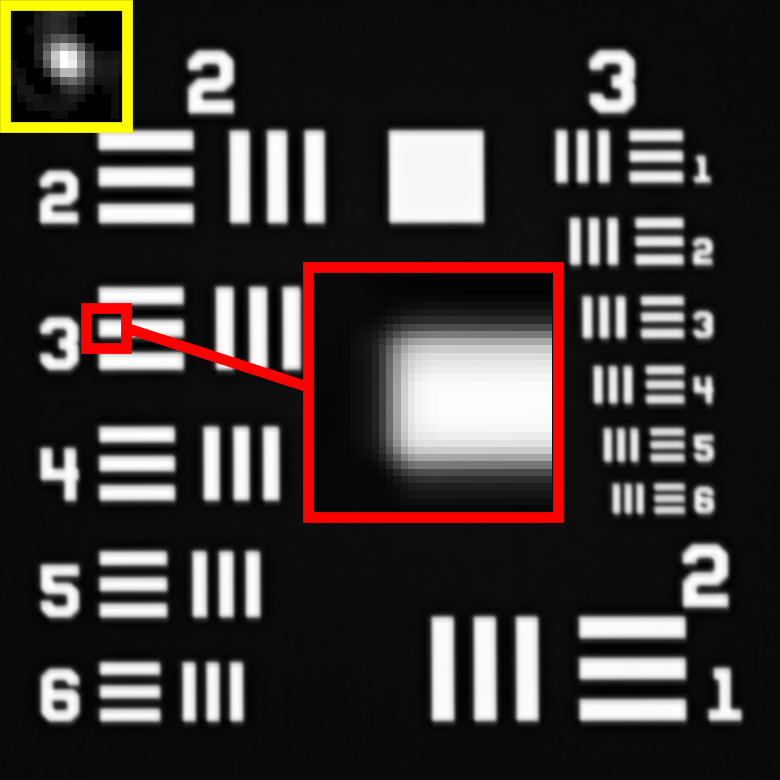} & \hspace{-2ex}
		\includegraphics[width=0.3\linewidth]{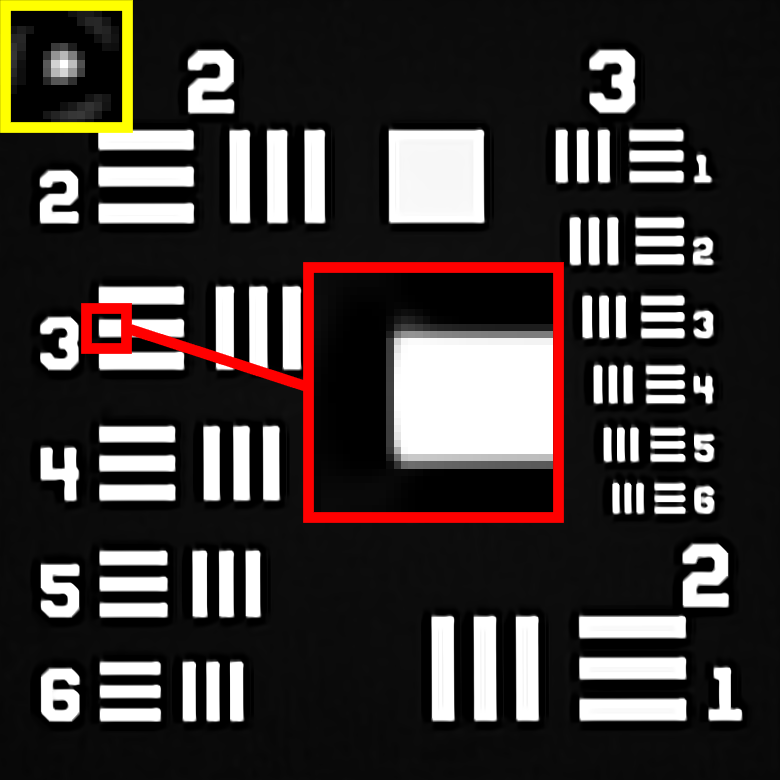} \\
		\small{(a) Ground Truth} & 
		\hspace{-2ex} \small{(b) Lucky Region} &
		\hspace{-2ex} \small{(c) Ours: 24.78 dB} \\
		\includegraphics[width=0.3\linewidth]{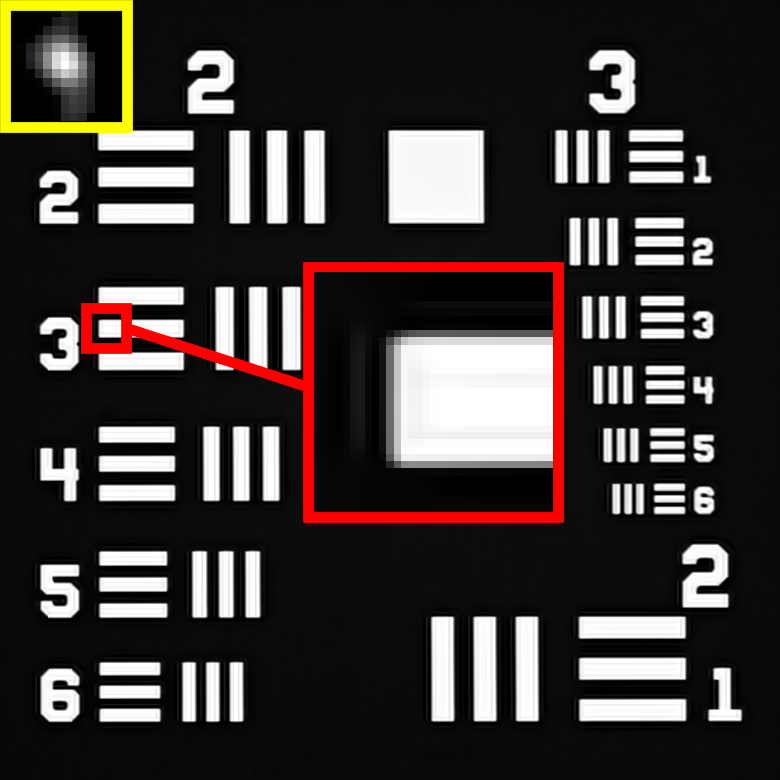} & \hspace{-2ex}
		\includegraphics[width=0.3\linewidth]{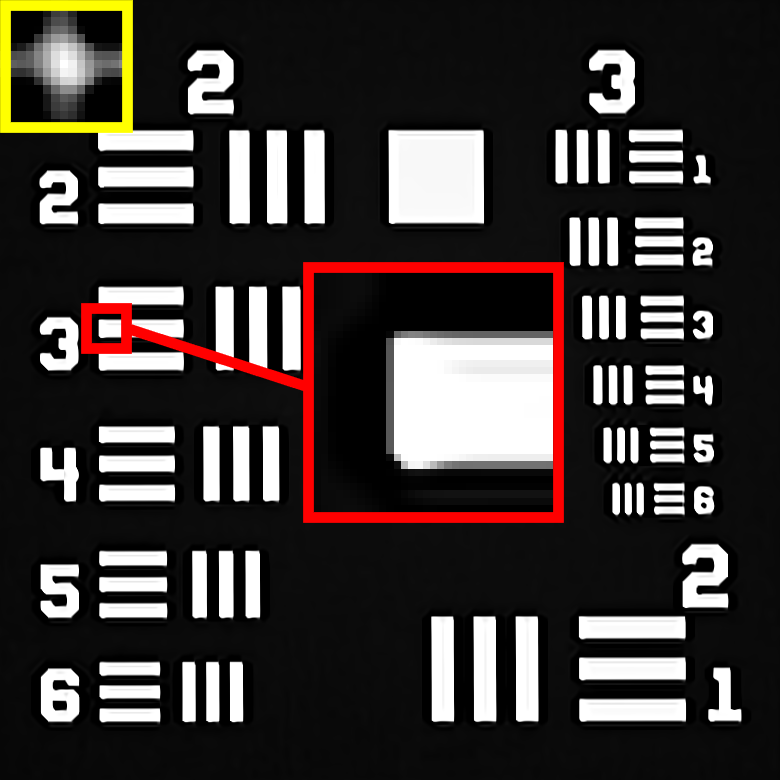} & \hspace{-2ex}
		\includegraphics[width=0.3\linewidth]{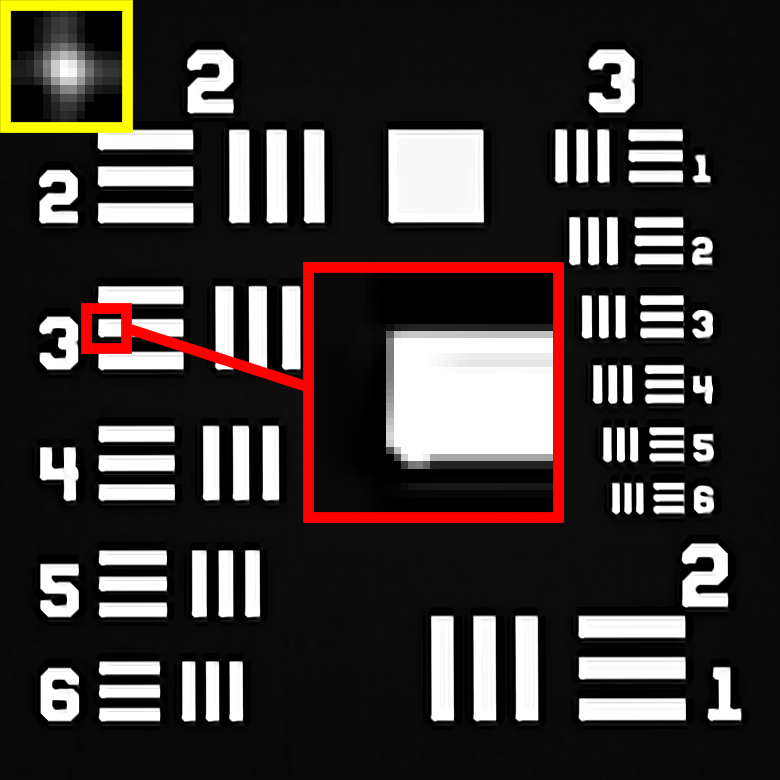} \\
		\small{(d) \cite{Shan2018}: 23.34 dB} &
		\hspace{-2ex} \small{(e) \cite{Chakrabarti2016}: 22.99 dB} & 
		\hspace{-2ex} \small{(f) \cite{Xu_2018}: 23.75 dB} \\
	\end{tabular}
	\caption{Experiments with USAF resolution chart. Corresponding PSFs are shown in the uppper-left corner. Zoom in for better view. }
	\label{fig:USAFresult}
	\vspace{-2ex}
\end{figure}

\subsection{Overall Algorithm on Real Data}
The third experiment is to compare the proposed pipeline with other turbulence restoration methods. Since we have reported simulated results in the previous two subsections, here we report performance on real data.

The first two sets of comparisons are the \texttt{Building} in \fref{fig: exp house} and the \texttt{Chimney} in \fref{fig:real_static} \cite{Hirsch2010}. We compare with three methods: Sobolev gradient flow by Lou et al. \cite{Lou2013}, B-spline + deblurring by Zhu and Milanfar \cite{Milanfar2013}, and a wavelet enhancement method by Anantrasirichai et al. \cite{Anantrasirichai2013}. The implementation of the methods are provided by the original authors, and the internal parameters are tuned according to the best of our knowledge. There are a few other methods discussed in the introduction, but we were not able to obtain the reproducible source codes. 

The next two sets of comparison are the \texttt{Man} sequence as shown in \fref{fig:real_static}, and the \texttt{Car} sequence as shown in \fref{fig: car moving}. The walking man's motion is highly horizontal, hence it is easily ``washed out'' by existing methods. In contrast, the proposed method is able to preserve the man and generate a reliable reference. The car sequence is a considerably harder problem, as the car is moving towards the viewer. The proposed method, however, is able to stabilize the background fences while sharpens the licence plate.

\begin{figure}[t]
	\centering
	\begin{tabular}{ccc}
		\includegraphics[width=0.3\linewidth]{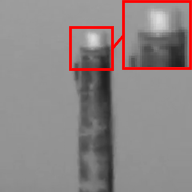}&
		\hspace{-2ex}\includegraphics[width=0.3\linewidth]{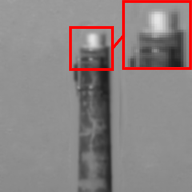}&
		\hspace{-2ex}\includegraphics[width=0.3\linewidth]{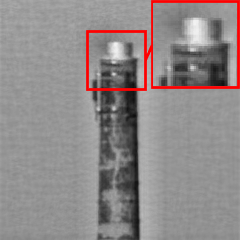}\\
		\small{(a) Input} & \hspace{-2ex} \small{(b) Lucky Region} & \hspace{-2ex} \small{(c) Anan. et al. \cite{Anantrasirichai2013}}\\
		\includegraphics[width=0.3\linewidth]{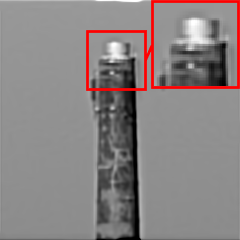}&
		\hspace{-2ex}\includegraphics[width=0.3\linewidth]{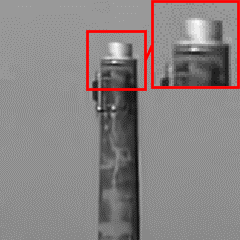}&
		\hspace{-2ex}\includegraphics[width=0.3\linewidth]{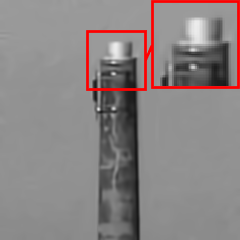}\\
		\small{(d) Lou et al. \cite{Lou2013}} & \hspace{-2ex} \small{(e) Zhu et al. \cite{Milanfar2013}} & \hspace{-2ex} \small{(f) Ours}
	\end{tabular}
	\vspace{-1ex}
	\caption{\texttt{Chimney}: A real static scene. The proposed method produces the sharpest recovery with minimal artifacts. Zoom-in for better visualization.}
	\label{fig:real_static}
	\vspace{-1ex}
\end{figure}

\begin{figure*}[t]
	\centering
	\begin{tabular}{cc}
		\hspace{-2ex} 
		\includegraphics[width=0.48\linewidth]{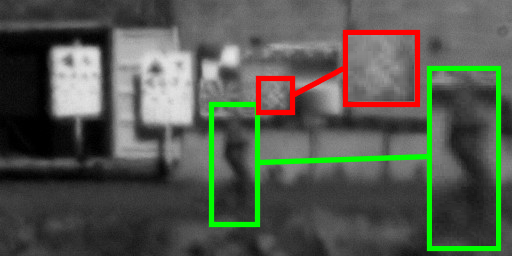} & 
		\hspace{-2ex} 
		\includegraphics[width=0.48\linewidth]{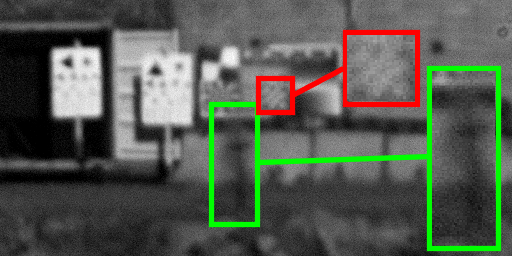}\\
		\hspace{-2ex} \small{(a) Raw Input} & 
		\hspace{-2ex} \small{(b) Lou et al. \cite{Lou2013}}\\
		\hspace{-2ex} 
		\includegraphics[width=0.48\linewidth]{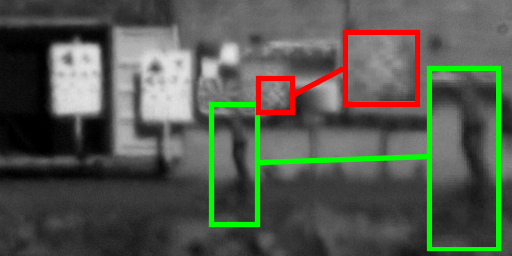} & 
		\hspace{-2ex} 
		\includegraphics[width=0.48\linewidth]{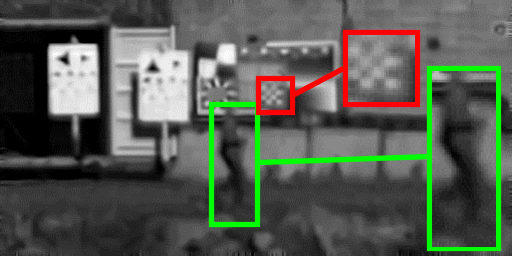}\\        
		\hspace{-2ex}\small{(c) Zhu et al. \cite{Milanfar2013}} & 
		\hspace{-2ex}\small{(d) Ours}
	\end{tabular}
	\vspace{-1ex}
	\caption{\texttt{Man}: A real dynamic scene. The proposed method preserves the moving object while stabilizing the background turbulence. Zoom-in for better visualization.}
	\label{fig:real_moving}
	\vspace{-1ex}
\end{figure*}

\begin{figure}[h]
	\centering
	\includegraphics[width=\linewidth]{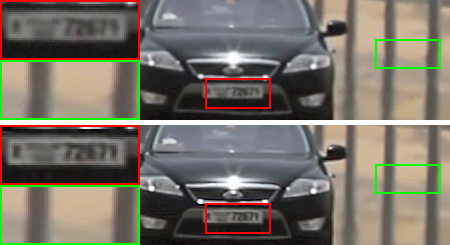}
	\vspace{-3ex}
	\caption{\texttt{Car}: A real dynamic scene. [Top] Raw input; Notice the blur and vibration of the background fences. [Bottom] Recovered by our proposed method. Zoom-in for better visualization.}
	\label{fig: car moving}
	\vspace{-2ex}
\end{figure}

\section{Conclusion}
We studied the image restoration pipeline of an atmospheric turbulence problem by grounding the design parameters on the physics of the turbulence. We showed that the non-local weight parameter $\beta$ should scale with the turbulence strength $D/r_0$. We proved that the ratio between the shift of the PSF $\sigma$ and the PSF bandwidth $\vnu$ is upper bounded by a constant. We demonstrated how a simple prior can outperform state-of-the-art blind deconvolution algorithms in the turbulence pipeline.


%

\appendices
\section{Proofs}
In this section, we clarify some theoretical details in the main paper, and prove theorem 1 and proposition \ref{prop: boxcar kernel sigma over nu}. We only provide a sketch of the proof of theorem \ref{maximum attained for lipschitz and continuous} here due to space limitations, and leave the complete proof of it to the supplementary notes.

For convenience, we denote in what follows the PDF and CDF of a Normal distribution with zero mean and $\sigma^2$ variance as $p_{\Theta_{\sigma}}$ and $P_{\Theta_{\sigma}}$ respectively.

Recall in definition \ref{def:short-psf from smoothing kernel} we restrict the short-PSFs considered to be those defined with smoothing kernels. We provide the definition of the latter here.

\begin{definition}[Smoothing Kernels]
We say a function $K:\mathbb{R}\to\mathbb{R}$ is a smoothing kernel if
\begin{enumerate}
    \item $0\le K(x) \le M$ everywhere on $\mathbb{R}$;
    \item $K$ is an even function;
    \item $\int_{\mathbb{R}}K(x)dx < \infty$.
\end{enumerate}
\label{def:appendix smoothing kernel}
\end{definition}

\noindent\textbf{Remark}. Our definition is inspired from, but more general than smoothing kernels studied in the theory of kernel density estimation from the probability and statistics literature, as this allows us to model more generic short-PSFs. For an introduction to the theory of kernel density estimation, see, for example, \cite{wasserman2004statistics}.

\begin{proof}[Proof of Theorem 1]
From Bernstein's inequality stated in lemma \ref{lemma: bernstein ineq} below, we have:
\begin{equation} 
\begin{aligned}
    &\mathbb{P}\left(\left| \widetilde{h}_{\nu}(x) - (h_{\nu} \circledast p_{\Theta})(x) \right| > \epsilon  \right) \\
    &\quad \le \, 2\exp\left\{-\frac{\epsilon^2 T}{2 V_{\nu}(x,\sigma)+2\nu^{-1}M\epsilon/3}\right\}.
\end{aligned}
\end{equation}
Now apply $\sup_{x\in\mathbb{R}}$ to both sides of the inequality, the left hand side is in the form we desire, and for the right hand side, based on the relationship between continuous increasing functions and the supremum operation, we have the desired expression:
\begin{equation}
\begin{aligned}
& \sup_{x\in\mathbb{R}}2\exp\left\{-  \frac{T \epsilon^2}{2 V_{\nu}(x,\sigma) + 2M\nu^{-1}\epsilon/3}\right\} \\
&\quad = \, 2\exp\left\{-  \frac{T \epsilon^2}{2 \sup_{x\in\mathbb{R}}V_{\nu}(x,\sigma) + 2M\nu^{-1}\epsilon/3}\right\}.
\end{aligned}
\end{equation}
\end{proof}

\begin{lemma}(Bernstein's Inequality \cite{bernstein1946})

Let $\{X_j\}_{j=1}^n$ be a collection of independent random variables. Assume $\mathbb{P}(|X_j|\le M)=1$ for every $j$, and let $\Var[X_j]$ be finite for every $j$. Then for any $\epsilon>0$,
\begin{equation}
\begin{aligned}
    & \mathbb{P}\left(\left|\frac{1}{n}\sum_{j=1}^n (X_j - \E[X_j])\right|>\epsilon\right) \\
    \le & 2\exp\left\{-\frac{n\epsilon^2}{2\frac{1}{n}\sum_{j=1}^n \Var[X_j] + \frac{2}{3}M\epsilon}\right\}.
\end{aligned}
\end{equation}
If the random variables are further assumed to be identically distributed, and denoting $\E[X_j]=\mu$ and $\Var[X_j]=\sigma^2$, then the above inequality simplifies to the following:
\begin{equation}
    \mathbb{P}\left(\left|\frac{1}{n}\sum_{j=1}^n X_j - \mu\right|>\epsilon\right) \le 2\exp\left\{-\frac{n\epsilon^2}{2 \sigma^2 + \frac{2}{3}M\epsilon}\right\}.
\end{equation}
\label{lemma: bernstein ineq}
\end{lemma}

\begin{proof}[Proof of Proposition \ref{prop: boxcar kernel sigma over nu}]
Recall by definition that $h_{\nu}(x)=\frac{1}{\nu}K(\frac{x}{\nu})$, so $h_{\nu}(x)=\frac{1}{2\nu}$ when $x\in[-\nu,\nu]$, and vanishes everywhere else. We examine $V_{\nu}(x=0,\sigma)$:
\begin{equation*}
\begin{aligned}
    V_{\nu}(0,\sigma) 
    = & \int_{\mathbb{R}} h_{\nu}^2(0-\theta) p_{\Theta_{\sigma}}(\theta)d\theta - \left(\int_{\mathbb{R}} h_{\nu}(0-\theta) p_{\Theta_{\sigma}}(\theta)d\theta \right)^2 \\
    = & \int_{-\nu}^{\nu} \frac{1}{4\nu^2} p_{\Theta_{\sigma}}(\theta)d\theta - \left(\int_{-\nu}^{\nu} \frac{1}{2\nu} p_{\Theta_{\sigma}}(\theta)d\theta\right)^2 \\
    = & \frac{1}{4\nu^2} [P_{\Theta_{\sigma}}(\nu) - P_{\Theta_{\sigma}}(-\nu)] - \\
    &\quad \frac{1}{4\nu^2}[P_{\Theta_{\sigma}}(\nu) - P_{\Theta_{\sigma}}(-\nu)]^2
\end{aligned}
\end{equation*}
We use the common notation of $\Phi$ and $\phi$ as the cdf and pdf of the standard normal distribution respectively. Then the above expression can be rewritten as
\begin{equation*}
\begin{aligned}
    V_{\nu}(0,\sigma)
    = & \frac{1}{4\nu^2} \left[\Phi\left(\frac{\nu}{\sigma}\right) - \Phi\left(\frac{-\nu}{\sigma}\right)\right] - \\
    &\quad \frac{1}{4\nu^2}\left[\Phi\left(\frac{\nu}{\sigma}\right) - \Phi\left(\frac{-\nu}{\sigma}\right)\right]^2 \\
    = &  \frac{1}{4\nu^2}  \left[2\Phi\left(\frac{\nu}{\sigma}\right) - 1\right] - \frac{1}{4\nu^2}\left[2\Phi\left(\frac{\nu}{\sigma}\right) - 1\right]^2.
\end{aligned}
\end{equation*}

We wish to find the interval of $\sigma$ inside $[0,\infty)$ on which $V_{\nu}$ is increasing. We note a preliminary fact that $V_{\nu}(x=0,\sigma=0)=0$, and nonnegative for $\sigma\in(0,\infty)$. 

Take derivative of the above expression with respect to $\sigma$, and denoting $g_{\nu}(\sigma)=2\Phi\left(\frac{\nu}{\sigma}\right) - 1$, we obtain:
\begin{equation}
    \frac{\partial V_{\nu}}{\partial\sigma}(0,\sigma) = \frac{1}{4\nu^2}\frac{\partial g_{\nu}}{\partial\sigma}(\sigma)\left[1-2g_{\nu}(\sigma)\right]
\end{equation}
More explicitly, using chain rule, we have
\begin{equation}
    \frac{\partial g_{\nu}}{\partial \sigma}(\sigma) = -\frac{2\nu}{\sigma^2}\phi\left(\frac{\nu}{\sigma}\right)
\end{equation}
Note that this expression is negative for every $\sigma\in(0,\infty)$. So to find the interval on which $V_{\nu}$ is increasing, we require the following to be true
\begin{equation}
\begin{aligned}
    & 1-2\left(2\Phi\left(\frac{\nu}{\sigma}\right) - 1\right) = 1-2g_{\nu}(\sigma) \le 0 \\
    & \iff \frac{1}{2} \le 2\Phi\left(\frac{\nu}{\sigma}\right) -1 \\
    & \iff \Phi^{-1}\left(\frac{3}{4}\right) \le \frac{\nu}{\sigma}.
\end{aligned}
\end{equation}
Since $\Phi^{-1}\left(\frac{3}{4}\right)\approx 0.67$, it holds that $\sigma\le 1.48\nu$. Knowing that $V_{\nu}(0,\sigma)$ is continuous at $0$ (simple application of dominated convergence), it follows that on the interval $[0,1.48\nu]$, $V_{\nu}(0,\sigma)$ is increasing.
\end{proof}
\noindent\textbf{Remark}. 
The reason we only consider $V_{\nu}(0,\sigma)$ instead of $\sup_{x\in\mathbb{R}}V_{\nu}(x,\sigma)$ for the boxcar kernel case is that, from experiments we observed that for $\nu$ and $\sigma$ that is not very small, the maximizer of $V_{\nu}(x,\sigma)$ with respect to $x$ is very close to $x=0$, and for $\nu$ that is not small, $\sup_{x\in\mathbb{R}}V_{\nu}(x,\sigma)$ does not attain a global maximum at very small $\sigma$'s.

\begin{proof}[Sketch of proof of theorem \ref{maximum attained for lipschitz and continuous}]
There are essentially two things that we need to show: 
\begin{enumerate}
    \item The function $\sup_{x\in\mathbb{R}}V_{\nu}(x,\sigma)$ is continuous as a function of $\sigma$ on $[0,\infty)$;
    \item $\sup_{x\in\mathbb{R}}V_{\nu}(x,\sigma)$ tends to $0$ when $\sigma\to\infty$
\end{enumerate}
If the above two conditions hold, then with a standard argument using compact sets to exhaust $[0,\infty)$ and studying the behavior of $\sup_{x\in\mathbb{R}}V_{\nu}(x,\sigma)$ on these compact sets, the desired result follows.

Looking at the expression of $\sup_{x\in\mathbb{R}}V_{\nu}(x,\sigma)$:
\begin{equation}
\begin{aligned}
    \sup_{x\in\mathbb{R}}V_{\nu}(x,\sigma)
    = & \sup_{x\in\mathbb{R}}\biggl\{ \int_{\mathbb{R}} h_{\nu}^2(x-\theta) p_{\Theta_{\sigma}}(\theta)d\theta - \\
    & \left(\int_{\mathbb{R}} h_{\nu}(x-\theta) p_{\Theta_{\sigma}}(\theta)d\theta \right)^2 \biggr\}
\end{aligned}
\end{equation}
We note that, to study $\lim_{\sigma\to\sigma_0}\sup_{x\in\mathbb{R}}V_{\nu}(x,\sigma)$ for $\sigma_0\in[0,\infty]$, the basic argument of relying on the weak convergence of probability measures using the dominated convergence theorem does not really work here due to the presence of the $\sup_{x\in\mathbb{R}}$ term.

To resolve the above issue, we rely on the following classical analysis result: 
\begin{lemma}
Suppose we have a sequence of functions $\{g_n:\mathbb{R}\to\mathbb{R}\}_{n=1}^{\infty}$ that converges uniformly to a function $g:\mathbb{R}\to\mathbb{R}$, and $\sup_{x\in\mathbb{R}}g(x)<\infty$. Then the following is true:
\begin{equation}
    \lim_{n\to\infty} \sup_{x\in\mathbb{R}} g_n(x) = \sup_{x\in\mathbb{R}}g(x)
\end{equation}
\end{lemma}
It follows from the above lemma that, $\sup_{x\in\mathbb{R}}V_{\nu}(x,\sigma)$ is continuous at $\sigma_0 \in[0,\infty)$ if for any sequence $\{\sigma_n\}_{n=1}^{\infty}$ satisfying $\sigma_n\to\sigma_0$, $V_{\nu}(x,\sigma_n)$ tends \textit{uniformly} to $V_{\nu}(x,\sigma_0)$; similar arguments can be applied to showing $\lim_{\sigma\to\infty}\sup_{x\in\mathbb{R}}V_{\nu}(x,\sigma)=0$. However, showing uniform convergence does require more global assumptions on the function $h_{\nu}$ in addition to those stated in definition \ref{def:appendix smoothing kernel}, namely Lipschitz continuity or continuity with compact support. We leave elaborations on these technical details to the supplementary notes.
\end{proof}


\section*{Acknowledgment}
The work is funded, in part, by the Air Force Research Lab and Leidos. The authors would like to thank Michael Rucci, Barry Karch, Daniel LeMaster and Edward Hovenac of the Air Force Research Lab for many insightful discussions. The authors also thank Nantheera Anantrasirichai, a co-author of \cite{Anantrasirichai2013}, and Peyman Milanfar, a co-author of \cite{Milanfar2013}, for generously sharing their MATLAB code.

This work has been cleared for public release carrying the approval number 88ABW-2019-2438.

\ifCLASSOPTIONcaptionsoff
  \newpage
\fi

\bibliographystyle{IEEEtran}
\bibliography{egbib}

%









\end{document}


\title{Supplementary Materials}

\maketitle

\begin{abstract}
This supplementary report provides the following additional information of the main article.
\begin{itemize}
\item Clarifications and proofs for theoretical results in section 2
\item Turbulence simulator
\end{itemize}
\end{abstract}

\section{Clarifications and Proofs for Theoretical Results in Section 2}
In this section, we provide the proof of theorem 2 from the main paper.

\subsection{Introduction}
The proof of theorem 2 is new. The closest known result is in the theory of kernel density estimation, however, a large portion of the efforts is devoted to the study of various probabilistic properties of the estimator (including its variance) with respect to the design parameters of the smoothing kernel, not the variance of the underlying distribution. Moreover, as stated in the appendix of the main paper, the proof really boils down to showing the following two results: continuity of the function $\sup_{x\in\mathbb{R}}V_{\nu}(x,\sigma)$ and it approaching $0$ when $\sigma\to\infty$. They both hinge on promoting the usual weak convergence of the probability measures to the more ``uniform'' form we desire. This is done by enforcing more global assumptions than, for example, bounded continuity on the smoothing kernel itself.

\subsection{Review of Notation and Assumptions}
We let $\Theta\sim\calN(0,\sigma^2)$, and $\{\theta_t\}_{t=1}^T$ iid copies of $\Theta$. We denote the probability density function of $\Theta$ as $p_{\Theta}$, and its cumulative distribution function as $P_{\Theta}$. To emphasize the dependence of $\Theta$ on its variance $\sigma^2$, we will sometimes write $\Theta_{\sigma}$ and correspondingly $p_{\Theta_{\sigma}}$ (or $p_{\sigma}$) and $P_{\Theta_{\sigma}}$ (or $P_{\sigma}$) for the pdf and cdf of $\Theta$ in this section.


With the definition of smoothing kernel from the main paper, we can define the short-PSF $h_{\nu}$ as follow.

\begin{assumption}
We assume that the short-PSF $h(\cdot)$ takes the form of
\begin{equation}
    h_{\nu}(x) = \frac{1}{\nu}K\left(\frac{x}{\nu}\right)
\end{equation}
for some smoothing kernel $K$, and some constant $\nu>0$.
\end{assumption}

Finally, we restate some notations established in the paper.
\begin{enumerate}
    \item We denote the finite sample estimate of the long-PSF as
    \begin{equation}
        \widetilde{h}_{\nu}(x) = \frac{1}{T} \sum_{t=1}^{T} h_{\nu}(x-\theta_t)
    \end{equation}
    
    \item We denote the true long-PSF as
    \begin{equation}
        \E_{\Theta}[h_{\nu}(x-\Theta)] = (h_{\nu} \circledast p_{\Theta})(x).
    \end{equation}
    
    \item To emphasize the dependence of $\Var_{\Theta_{\sigma}}[h(x-\Theta_{\sigma})]$ on $\sigma$, and for simplicity of notation, we write
    \begin{equation}
        V_{\nu}(x,\sigma) = \Var_{\Theta_{\sigma}}[h_{\nu}(x-\Theta_{\sigma})]
    \end{equation}
\end{enumerate}

\subsection{Proof of Theorem 2}
We restate the theorem first, and provide a complete proof of it.

\noindent \textbf{Theorem 2}. \textit{Under the assumptions made in the previous section, and if $K$ is further assumed to be either Lipschitz continuous on $\mathbb{R}$, or continuous but with compact support on $\mathbb{R}$, then $\sup_{x\in\mathbb{R}}V_{\nu}(x,\sigma)$ is \textbf{continuous in $\sigma$} on $[0,\infty)$ and tends to $0$ as $\sigma\to\infty$, therefore it \textbf{attains a global maximum} (not necessarily unique) on $(0,\infty)$.}
\begin{proof}
We wish to show the following two results:
\begin{enumerate}
    \item The function $\sup_{x\in\mathbb{R}}V_{\nu}(x,\sigma)$ is continuous as a function of $\sigma$ on $[0,\infty)$;
    \item $\sup_{x\in\mathbb{R}}V_{\nu}(x,\sigma)$ tends to $0$ when $\sigma\to\infty$
\end{enumerate}
We claim that if the above two conditions hold, then $\sup_{x\in\mathbb{R}}V_{\nu}(x,\sigma)$ has to attain a global maximum on $(0,\infty)$. We prove this claim first. First of all, note the preliminary facts that $\sup_{x\in\mathbb{R}}V_{\nu}(x,\sigma)$ is nonnegative on $[0,\infty)$, equal to $0$ when $\sigma=0$ (since the distribution is degenerate at this point), and is not identically $0$ on all of $[0,\infty)$. Now, define the infinite increasing sequence of compact sets that exhausts $[0, \infty)$: $\{E_n=[0,n]\}_{n=1}^{\infty}$. Following from 1., since $\sup_{x\in\mathbb{R}}V_{\nu}(x,\sigma)$ is continuous on each $E_n$, we observe the following: first, the number $v_n = \sup_{\sigma\in E_n}\sup_{x\in \mathbb{R}}V_{\nu}(x,\sigma)$ is well-defined for every $n$ (i.e. they exist and are finite), second, the sequence of numbers $\{v_n\}_{n=1}^{\infty}$ has to be non-decreasing, third, there must exist some $N$ such that for every $n>N$, $v_n > 0$. Now we assume toward a contradiction that $\sup_{x\in \mathbb{R}}V_{\nu}(x,\sigma)$ does not attain a global maximum on $(0,\infty)$. This implies that, for every $n$, there exists some $n'>n$ such that $v_n < v_{n'}$. Then either the sequence $v_n\to\infty$ as $n\to\infty$, or there exists some number $U>0$ such that $v_n\to U$ as $n\to\infty$ with $v_n < U$ for every $n$. But any one of these cases contradicts 2., which completes the proof of our claim.

To prove 1. and 2., we require several lemmas from classical analysis. The following is a table of content of the proofs, which will be presented afterwards.
\begin{itemize}
    \item General preliminary results from analysis, stated with references
    \begin{itemize}
        \item Lemma \red{2}, which provides a sufficient condition for swapping limits and suprema:
        \begin{equation}
            \lim_{n\to\infty} \sup_{x\in\mathbb{R}} g_n(x) = \sup_{x\in\mathbb{R}} \lim_{n\to\infty} g_n(x)
        \end{equation}
        \item Lemma \ref{lemma: general dct}, the general dominated convergence theorem
        \item Lemma \ref{lemma: density of lipschitz}, density of the space of Lipschitz continuous functions in the space of continuous functions, assuming a common compact metric space as their domains
    \end{itemize}
    \item Proof of 1.
    \begin{itemize}
        \item Corollary \ref{corollary: continuity on [0,infty)}, which follows from the following three lemmas
        \begin{itemize}
            \item Lemma \ref{lemma: continuity on (0,infty)}, which proves the continuity of $\sup_{x\in\mathbb{R}}V_{\nu}(x,\sigma)$ on the open ray $(0,\infty)$. It makes use of lemmas \red{2} and \ref{lemma: general dct}
            \item Lemma \ref{lemma: continuity at 0, lipschitz}, which proves the continuity of $\sup_{x\in\mathbb{R}}V_{\nu}(x,\sigma)$ at $\sigma=0$ assuming that $h_{\nu}$ the short-PSF is Lipschitz continuous. It makes use of lemma \red{2}
            \item Lemma \ref{lemma: continuity at 0, continuous compact support},  which also proves the continuity of $\sup_{x\in\mathbb{R}}V_{\nu}(x,\sigma)$ at $\sigma=0$, but this time assuming $h_{\nu}$ is continuous with compact support on $\mathbb{R}$. It makes use of lemmas \red{2}, \ref{lemma: density of lipschitz} and \ref{lemma: continuity at 0, lipschitz}
        \end{itemize}
    \end{itemize}
    \item Proof of 2.
    \begin{itemize}
        \item Lemma \ref{lemma: tend to 0 as sigma tend to infty}
    \end{itemize}
\end{itemize}
\end{proof}

\begin{definition}
Given a sequence of functions $\{g_n:\calX\subseteq\mathbb{R}\to\mathbb{R}\}_{n=1}^{\infty}$. We say the sequence \textbf{converges uniformly} to a function $g:\calX\to\mathbb{R}$ if for every $\epsilon>0$, there exists $n'\in\mathbb{N}$ such that for all $n\ge n'$, 
\begin{equation}
    \sup_{x\in\calX} |g_n(x)-g(x)|<\epsilon
\end{equation}
An equivalent definition often seen is the following: the sequence $g_n$ converges uniformly to $g$ if
\begin{equation}
    \lim_{n\to\infty} \sup_{x\in\calX} |g_n(x)-g(x)|=0
\end{equation}
\end{definition}

\begin{lemma}[General Dominated Convergence Theorem \cite{royden2010}]
Let $E$ be a Lebesgue measurable set. Let $\{f_n\}_{n=1}^{\infty}$ be a sequence of Lebesgue measurable functions on $E$, and let $\{g_n\}_{n=1}^{\infty}$ be a sequence of nonnegative Lebesgue measurable functions on $E$. Further assume the following:
\begin{enumerate}
    \item $\lim_{n\to\infty}f_n(x)=f(x)$ almost everywhere on $E$
    \item $\lim_{n\to\infty}g_n(x)=g(x)$ almost everywhere on $E$
    \item $|f_n(x)|\le g_n(x), \forall x \in E, \forall n$
    \item $\lim_{n\to\infty}\int_E g_n = \int_E g < \infty$
\end{enumerate}
Then $\lim_{n\to\infty}\int_E f_n = \int_E f$.
\label{lemma: general dct}
\end{lemma}

\begin{lemma}[Density of $Lip(\calX)$ in $C(\calX)$, $\calX$ compact \cite{carothers2010}]
Fix a compact metric space $\calX$. For any $\epsilon>0$ and any continuous function $f:\calX\to\mathbb{R}$, there exists a Lipschitz continuous function $g:\calX\to\mathbb{R}$ such that $\sup_{x\in\calX}|f(x)-g(x)|<\epsilon$.
\label{lemma: density of lipschitz}
\end{lemma}

    

\begin{lemma}
The function $\sup_{x\in\mathbb{R}}V_{\nu}(x,\sigma)$ is continuous in $\sigma$ on $(0,\infty)$.
\label{lemma: continuity on (0,infty)}
\end{lemma}
\begin{proof}
When necessary, we write explicitly $p_{\sigma}(x) = \frac{1}{\sqrt{2\pi}\sigma} \exp\big\{-\frac{x^2}{2\sigma^2}\big\}$.

By the sequence definition of continuity, it suffices to show the following: fix any $\sigma_0 \in (0,\infty)$, for any sequence $\{\sigma_n\}_{n=1}^{\infty}\subset(0,\infty)$ that satisfies $\lim_{n\to\infty}\sigma_n=\sigma_0$ with $\sigma_n\neq\sigma_0$ for all $n$, the following is true 
\begin{equation}
    \lim_{n\to\infty} \sup_{x\in\mathbb{R}} V_{\nu}(x,\sigma_n) = \sup_{x\in\mathbb{R}} V_{\nu}(x,\sigma_0)
\end{equation}

It then follows from lemma \red{2} that it is sufficient for us to show the sequence of functions in $x$, $\{V_{\nu}(x,\sigma_n)\}_{n=1}^{\infty}$, tends uniformly to the function $V_{\nu}(x,\sigma_0)$, and that $\sup_{x\in\mathbb{R}}V_{\nu}(x,\sigma_0)<\infty$ must be true.

So choose any $\{\sigma_n\}_{n=1}^{\infty}$ satisfying properties aforementioned. Uniform boundedness of $V_{\nu}(x,\sigma)$ is can be shown as follows: for any $x\in\mathbb{R}$ and $\sigma\in(0,\infty)$,
\begin{equation}
\begin{aligned}
    |V_{\nu}(x,\sigma)| 
    = & V_{\nu}(x,\sigma) \\
    \le & \int_{\mathbb{R}} \frac{1}{\nu^2}K^2\left(\frac{x-\theta}{\nu}\right) p_{\sigma}(\theta)d\theta\\
    \le & \frac{1}{\nu^2} M^2 \int_{\mathbb{R}} p_{\sigma}(\theta)d\theta \\
    = & \frac{1}{\nu^2} M^2
\end{aligned}
\end{equation}

Now we show uniform convergence. We first observe a useful property of $V_{\nu}$. Fix any $\sigma_0 \in (0,\infty)$, for any $x\in\mathbb{R}$, any $n$, the following holds
\begin{equation}
\begin{aligned}
    |V_{\nu}(x,\sigma_n) - V_{\nu}(x,\sigma_0)| 
    = & \Big|\int_{\mathbb{R}} \frac{1}{\nu^2}K^2\left(\frac{x-\theta}{\nu}\right) p_{\sigma_n}(\theta)d\theta - \Big(\int_{\mathbb{R}} \frac{1}{\nu}K\left(\frac{x-\theta}{\nu}\right) p_{\sigma_n}(\theta)d\theta\Big)^2 \\
    & - \int_{\mathbb{R}} \frac{1}{\nu^2}K^2\left(\frac{x-\theta}{\nu}\right) p_{\sigma_0}(\theta)d\theta + \Big(\int_{\mathbb{R}} \frac{1}{h}K\left(\frac{x-\theta}{\nu}\right) p_{\sigma_0}(\theta)d\theta\Big)^2\Big| \\
    \le & \left|\int_{\mathbb{R}} \frac{1}{\nu^2}K^2\left(\frac{x-\theta}{\nu}\right) (p_{\sigma_n}(\theta) - p_{\sigma_0}(\theta))d\theta \right| + \\
    & \left|\Big(\int_{\mathbb{R}} \frac{1}{\nu}K\left(\frac{x-\theta}{\nu}\right) p_{\sigma_n}(\theta)d\theta\Big)^2 - \Big(\int_{\mathbb{R}} \frac{1}{\nu}K\left(\frac{x-\theta}{\nu}\right) p_{\sigma_0}(\theta)d\theta\Big)^2 \right| \\
    \le & \frac{1}{\nu^2}M^2\int_{\mathbb{R}}|p_{\sigma_n}(\theta)-p_{\sigma_0}(\theta)|d\theta + \\
    & \left|\int_{\mathbb{R}} \frac{1}{\nu}K\left(\frac{x-\theta}{\nu}\right) (p_{\sigma_n}(\theta)+p_{\sigma_0}(\theta))d\theta\right| 
    \left|\int_{\mathbb{R}} \frac{1}{\nu}K\left(\frac{x-\theta}{\nu}\right) (p_{\sigma_n}(\theta)-p_{\sigma_0}(\theta))d\theta\right| \\
    \le & \frac{1}{\nu^2}M^2\int_{\mathbb{R}}|p_{\sigma_n}(\theta)-p_{\sigma_0}(\theta)|d\theta + \\
    & \Big(\frac{1}{\nu}M\int_{\mathbb{R}}p_{\sigma_n}(\theta)+p_{\sigma_0}(\theta)d\theta\Big)
    \Big(\frac{1}{\nu}M\int_{\mathbb{R}} |p_{\sigma_n}(\theta)-p_{\sigma_0}(\theta)|d\theta\Big) \\
    = & 3\frac{1}{\nu^2}M^2 \int_{\mathbb{R}} |p_{\sigma_n}(\theta)-p_{\sigma_0}(\theta)|d\theta \\
    = & U(\sigma_n)
\end{aligned}
\end{equation}
Notice that this upper bound is independent of $x$, so we can show the sequence $\{V(x,\sigma_n)\}_{n=1}^{\infty}$ converges uniformly to $V(x,\sigma_0)$ if the above upper bound converges to $0$ when $n\to\infty$.

\vspace{2ex}
We make use of lemma \ref{lemma: general dct}, the general dominated convergence theorem. Using the notation of that lemma, we let $f_n(\theta) = |p_{\sigma_n}(\theta)-p_{\sigma_0}(\theta)|$, and $g_n(\theta) = p_{\sigma_n}(\theta) + p_{\sigma_0}(\theta)$. A few remarks:
\begin{enumerate}
    \item $\lim_{n\to\infty}f_n(\theta)=0, \forall \theta\in\mathbb{R}$, because the function is continuous in $\sigma$
    \item $\lim_{n\to\infty}g_n(\theta)=2 p_{\sigma_0}(\theta) = g(\theta), \forall \theta\in\mathbb{R}$
    \item Clearly $|f_n(\theta)|\le g_n(\theta), \forall \theta\in\mathbb{R}$
    \item Since $\int_{\mathbb{R}}g_n(\theta)d\theta = 2$, and $\int_{\mathbb{R}}g(\theta)d\theta = 2$, the final requirement for general dominated convergence theorem is also satisfied
\end{enumerate}
It follows that
\begin{equation}
    \lim_{n\to\infty} U(\sigma_n) = 3\frac{1}{\nu^2}M^2\int_{\mathbb{R}}\lim_{n\to\infty}f_n(\theta)d\theta = 0
\end{equation}
This completes the proof.
\end{proof}

\begin{lemma}
If $K$ is further assumed to be Lipschitz continuous on $\mathbb{R}$, then $\sup_{x\in\mathbb{R}}V_{\nu}(x,\sigma)$ is continuous at $\sigma=0$. 
\label{lemma: continuity at 0, lipschitz}
\end{lemma}
\begin{proof}

Let's first denote the Lipschitz constant of $K$ as $L_K$.

It is sufficient, by lemma \red{2}, to prove that for any sequence $\sigma_n\to0$ with $\sigma_n\neq 0$ for all $n$, the sequence of functions $V_{\nu}(x,\sigma_n)$ converges uniformly to $0$. This task can be further divided into showing that the two terms in $V_{\nu}(x,\sigma_n)=\int_{\mathbb{R}} \frac{1}{\nu^2}K^2\left(\frac{x-\theta}{\nu}\right)\frac{1}{\sqrt{2\pi}\sigma_n}e^{-\frac{\theta^2}{2\sigma_n^2}}d\theta - \left(\int_{\mathbb{R}} \frac{1}{h}K\left(\frac{x-\theta}{h}\right)\frac{1}{\sqrt{2\pi}\sigma_n}e^{-\frac{\theta^2}{2\sigma_n^2}}d\theta\right)^2$ each converges uniformly to the same function of $x$.

\begin{enumerate}
    \item We start by considering the first term in $V_{\nu}(x,\sigma_n)$. We prove it converges uniformly to the function $1/\nu^2K^2(x/\nu)$: for all $x\in\mathbb{R}$, any $n$,
    \begin{equation}
    \begin{aligned}
        & \left|\int_{\mathbb{R}} \frac{1}{\nu^2}K^2\left(\frac{x-\theta}{\nu}\right)\frac{1}{\sqrt{2\pi}\sigma_n}e^{-\frac{\theta^2}{2\sigma_n^2}}d\theta - \frac{1}{\nu^2}K^2\left(\frac{x}{\nu}\right)\right| \\
        = & \left| \int_{\mathbb{R}} \frac{1}{\nu^2}\left(K^2\left(\frac{x-\theta}{\nu}\right)-K^2\left(\frac{x}{\nu}\right)\right)\frac{1}{\sqrt{2\pi}\sigma_n}e^{-\frac{\theta^2}{2\sigma_n^2}}d\theta\right| \\
        \le & \int_{\mathbb{R}} \frac{1}{\nu^2} \left|K\left(\frac{x-\theta}{\nu}\right)+K\left(\frac{x}{\nu}\right)\right|\left|K\left(\frac{x-\theta}{\nu}\right)-K\left(\frac{x}{\nu}\right)\right|\frac{1}{\sqrt{2\pi}\sigma_n}e^{-\frac{\theta^2}{2\sigma_n^2}}d\theta \\
        \le & \frac{2}{\nu^2}M \int_{\mathbb{R}} \left|K\left(\frac{x-\theta}{\nu}\right)-K\left(\frac{x}{\nu}\right)\right|\frac{1}{\sqrt{2\pi}\sigma_n}e^{-\frac{\theta^2}{2\sigma_n^2}}d\theta \\
        = & \frac{2}{\nu^2}M \int_{\mathbb{R}} \left|K\left(\frac{x-\sigma_n \gamma}{\nu}\right)-K\left(\frac{x}{\nu}\right)\right|\frac{1}{\sqrt{2\pi}}e^{-\frac{\gamma^2}{2}}d\gamma \quad \text{(Change of variable, $\gamma=\theta/\sigma_n$)} \\
        \le & \frac{2}{\nu^2}M \int_{\mathbb{R}} L_K\left|\frac{x-\sigma_n\gamma}{\nu}-\frac{x}{\nu}\right|\frac{1}{\sqrt{2\pi}}e^{-\frac{\gamma^2}{2}}d\gamma \\ 
        = & \sigma_n\frac{2 M L_K}{\nu^3} \int_{\mathbb{R}} \left|\gamma\right|\frac{1}{\sqrt{2\pi}}e^{-\frac{\gamma^2}{2}}d\gamma \\
        \le & \sigma_nC
    \end{aligned}
    \end{equation}
    for some positive constant $C$ independent of $x$ and $n$. The last inequality is true since the integral $\int_{\mathbb{R}} |\gamma|\frac{1}{\sqrt{2\pi}}e^{-\frac{\gamma^2}{2}}d\gamma$ is finite by a classical calculus result. More importantly, this upper bound is independent of $x$ and only linearly dependent on $\sigma_n$, implying that the sequence of functions in $x$, $$\int_{\mathbb{R}}\frac{1}{\nu^2}K^2\left(\frac{x-\theta}{\nu}\right)\frac{1}{\sqrt{2\pi}\sigma_n}e^{-\frac{\theta^2}{2\sigma_n^2}}d\theta$$
    converges uniformly to $1/\nu^2K^2(x/\nu)$.
    
    \item Now we examine the second term in $V_{\nu}(x,\sigma_n)$. Again we prove it converges uniformly to $1/\nu^2K^2(x/\nu)$: for any $x\in\mathbb{R}$,
    \begin{equation}
    \begin{aligned}
        & \left|\left(\int_{\mathbb{R}} \frac{1}{h}K\left(\frac{x-\theta}{h}\right)\frac{1}{\sqrt{2\pi}\sigma_n}e^{-\frac{\theta^2}{2\sigma_n^2}}d\theta\right)^2 - \frac{1}{\nu^2}K^2\left(\frac{x}{\nu}\right)\right| \\
        = & \left|\int_{\mathbb{R}} \frac{1}{\nu}\left(K\left(\frac{x-\theta}{\nu}\right)+K\left(\frac{x}{\nu}\right)\right)\frac{1}{\sqrt{2\pi}\sigma_n}e^{-\frac{\theta^2}{2\sigma_n^2}}d\theta\right|\left|\int_{\mathbb{R}} \frac{1}{\nu}\left(K\left(\frac{x-\theta}{\nu}\right)-K\left(\frac{x}{\nu}\right)\right)\frac{1}{\sqrt{2\pi}\sigma_n}e^{-\frac{\theta^2}{2\sigma_n^2}}d\theta\right| \\
        \le & \frac{2}{\nu^2}M \int_{\mathbb{R}} \left|K\left(\frac{x-\theta}{\nu}\right)-K\left(\frac{x}{\nu}\right)\right|\frac{1}{\sqrt{2\pi}\sigma_n}e^{-\frac{\theta^2}{2\sigma_n^2}}d\theta
    \end{aligned}
    \end{equation}
    so by exactly the same reasoning as in the previous part, we showed that the sequence of functions in $x$,
    $$\left(\int_{\mathbb{R}}\frac{1}{\nu}K\left(\frac{x-\theta}{\nu}\right)\frac{1}{\sqrt{2\pi}\sigma_n}e^{-\frac{\theta^2}{2\sigma_n^2}}d\theta\right)^2$$
    converges uniformly to $1/\nu^2K^2(x/\nu)$.
\end{enumerate}
By the above two points, we see that $V_{\nu}(x,\sigma_n)$ converges uniformly to $0$ as $n\to\infty$, so by lemma \red{2}, 
\begin{equation}
    \lim_{n\to\infty}\sup_{x\in\mathbb{R}}V_{\nu}(x,\sigma_n) = \sup_{x\in\mathbb{R}} 0 = 0
\end{equation}
which completes the proof.
\end{proof}

\begin{lemma}
If $K$ is further assumed to be continuous and compactly supported on $\mathbb{R}$, then $\sup_{x\in\mathbb{R}}V_{\nu}(x,\sigma)$ is continuous at $\sigma=0$.
\label{lemma: continuity at 0, continuous compact support}
\end{lemma}
\begin{proof}
Our strategy here is quite similar to that in the previous one, though we need to make use of lemma \ref{lemma: density of lipschitz} here. So again, note that it is sufficient, by lemma \red{2}, to prove that for any sequence $\sigma_n\to0$ with $\sigma_n\neq 0$ for all $n$, the sequence of functions $V_{\nu}(x,\sigma_n)$ converges uniformly to $0$. We again prove that the two terms in $V_{\nu}(x,\sigma_n) = \int_{\mathbb{R}} h_{\nu}^2(x-\theta)p_{\sigma_n}(\theta) d\theta - \left(\int_{\mathbb{R}} h_{\nu}(x-\theta)p_{\sigma_n}(\theta) d\theta\right)^2$ each converges uniformly to the same function of $x$, so $V_{\nu}(x,\sigma)$ converges uniformly to $0$. However, this time we need to set a few things up first.

Denote the compact support of $K$ by $\calX\subset\mathbb{R}$. First note that, by lemma \ref{lemma: density of lipschitz}, for any $\epsilon>0$, there exists a function $K_{\epsilon}$ that is Lipschitz continuous with support on $\calX$ such that $\sup_{x\in\mathbb{R}}|K_{\epsilon}(x)-K(x)|<\epsilon$.

Additionally, we denote $h_{\nu,\epsilon}(x) = \frac{1}{\nu}K_{\epsilon}\left(\frac{x}{\nu}\right)$. 

Now we can proceed. We treat the two terms in $V_{\nu}(x,\sigma)$ in order.

\begin{enumerate}
    \item Fix any $\epsilon>0$, choose compactly supported Lipschitz $K_{\epsilon}$ as described above. The following holds:
    \begin{equation}
    \begin{aligned}
        & \left|\int_{\mathbb{R}} h_{\nu}^2(x-\theta)p_{\sigma_n}(\theta) d\theta -h_{\nu}^2(x)\right| \\
        = & \left|\int_{\mathbb{R}} h_{\nu}^2(x-\theta) p_{\sigma_n}(\theta) d\theta - \int_{\mathbb{R}} h_{\nu,\epsilon}^2(x-\theta) p_{\sigma_n}(\theta) d\theta + \int_{\mathbb{R}} h_{\nu,\epsilon}^2(x-\theta) p_{\sigma_n}(\theta) d\theta - h_{\nu,\epsilon}^2(x) + h_{\nu,\epsilon}^2(x) -h_{\nu}^2(x)\right| \\
        = & \left|\int_{\mathbb{R}} (h_{\nu}^2 - h_{\nu,\epsilon}^2)(x-\theta)p_{\sigma_n}(\theta) d\theta + \int_{\mathbb{R}} \left(h_{\nu,\epsilon}^2(x-\theta) - h_{\nu,\epsilon}^2(x)\right)p_{\sigma_n}(\theta) d\theta + h_{\nu,\epsilon}^2(x) -h_{\nu}^2(x)\right| \\
        \le & \left|\int_{\mathbb{R}} (h_{\nu}^2 - h_{\nu,\epsilon}^2)(x-\theta)p_{\sigma_n}(\theta) d\theta\right| + \left|\int_{\mathbb{R}} \left(h_{\nu,\epsilon}^2(x-\theta) - h_{\nu,\epsilon}^2(x)\right)p_{\sigma_n}(\theta) d\theta \right| + \left|h_{\nu,\epsilon}^2(x) -h_{\nu}^2(x)\right| \\
        \le & \int_{\mathbb{R}} \left(\frac{2M+\epsilon}{\nu}\right)\left|h_{\nu} - h_{\nu,\epsilon}\right|(x-\theta)p_{\sigma_n}(\theta) d\theta + \int_{\mathbb{R}} \left(2\frac{M+\epsilon}{\nu}\right)\left|h_{\nu,\epsilon}(x-\theta) - h_{\nu,\epsilon}(x)\right|p_{\sigma_n}(\theta) d\theta \\
        & + \left(\frac{2M+\epsilon}{\nu}\right) \left|h_{\nu,\epsilon}(x) -h_{\nu}(x)\right| \\
        \le & 2\epsilon\left(\frac{2M+\epsilon}{\nu}\right) + \int_{\mathbb{R}} \left(2\frac{M+\epsilon}{\nu}\right)\left|h_{\nu,\epsilon}(x-\theta) - h_{\nu,\epsilon}(x)\right|p_{\sigma_n}(\theta) d\theta \\
        = & \epsilon C_1 + \left(2\frac{M+\epsilon}{\nu}\right) \int_{\mathbb{R}} \frac{1}{\nu}\left|K_{\epsilon}\left(\frac{x-\theta}{\nu}\right)-K_{\epsilon}\left(\frac{x}{\nu}\right)\right|\frac{1}{\sqrt{2\pi}\sigma_n}e^{-\frac{\theta^2}{2\sigma_n^2}}d\theta
    \end{aligned}
    \end{equation}
    where $C_1$ is a positive constant independent of $x$, $n$ (and is bounded when $\epsilon\to0$). Since $K_{\epsilon}$ is Lipschitz continuous, we apply the same argument as in the proof of the previous lemma, and arrive at
    \begin{equation}
        \left|\int_{\mathbb{R}} h_{\nu}^2(x-\theta)p_{\sigma_n}(\theta) d\theta -h_{\nu}^2(x)\right| \le \epsilon C_1 + \left(2\frac{M+\epsilon}{\nu}\right) \sigma_n C_2, \forall n
    \end{equation}
    where $C_2$ is some positive constant independent of $x$ and $n$, dependent on $\epsilon$. Then we can show uniform convergence if the RHS of the inequality converges to $0$ when $n\to\infty$. What we do have is the following:
    \begin{equation}
        \lim_{n\to\infty}\sup_{x\in\mathbb{R}}\left|\int_{\mathbb{R}} h_{\nu}^2(x-\theta)p_{\sigma_n}(\theta) d\theta -h_{\nu}^2(x)\right| \le \epsilon C_1, \forall \epsilon>0
    \end{equation}
    Push $\epsilon$ to $0$, we obtain uniform convergence.
    
    \item Now we examine the second term in $V_{\nu}(x,\sigma_n)$.
    \begin{equation}
    \begin{aligned}
        & \left|\left(\int_{\mathbb{R}} h_{\nu}(x-\theta)p_{\sigma_n}(\theta) d\theta\right)^2 -h_{\nu}^2(x)\right| \\
        \le & \int_{\mathbb{R}} \left|h_{\nu}(x-\theta) + h_{\nu}(x)\right|p_{\sigma_n}(\theta) d\theta \int_{\mathbb{R}} \left|h_{\nu}(x-\theta)- h_{\nu}(x)\right|p_{\sigma_n}(\theta) d\theta\\
        \le & 2\frac{1}{\nu}M \int_{\mathbb{R}} \left|h_{\nu}(x-\theta) -h_{\nu,\epsilon}(x-\theta)+h_{\nu,\epsilon}(x-\theta)-h_{\nu,\epsilon}(x)+h_{\nu,\epsilon}(x) - h_{\nu}(x)\right|p_{\sigma_n}(\theta) d\theta\\
        \le & \frac{2M}{\nu}\left(2\epsilon + \int_{\mathbb{R}} \left|h_{\nu,\epsilon}(x-\theta)- h_{\nu,\epsilon}(x)\right|p_{\sigma_n}(\theta) d\theta\right) \\
        = & \frac{2M}{\nu}\left(2\epsilon + \int_{\mathbb{R}} \frac{1}{\nu}\left|K_{\epsilon}\left(\frac{x-\theta}{\nu}\right)-K_{\epsilon}\left(\frac{x}{\nu}\right)\right|\frac{1}{\sqrt{2\pi}\sigma_n}e^{-\frac{\theta^2}{2\sigma_n^2}}d\theta\right)
    \end{aligned}
    \end{equation}
    Again, using the same argument as in the proof of the previous lemma, we have
    \begin{equation}
        \left|\left(\int_{\mathbb{R}} h_{\nu}(x-\theta)p_{\sigma_n}(\theta) d\theta\right)^2 -h_{\nu}^2(x)\right| \le \epsilon C_1 + \sigma_n C_2
    \end{equation}
    for positive constants $C_1, C_2$ independent of $x$ and $n$, and $C_2$ dependent on $\epsilon$. With a similar argument to the previous point, we can show uniform convergence.
\end{enumerate}
Continuity of $\sup_{x\in\mathbb{R}}V_{\nu}(x,\sigma)$ at $\sigma=0$ now follows.
\end{proof}

\begin{corollary}
If $K$ is Lipschitz continuous on $\mathbb{R}$, or continuous and compactly supported on $\mathbb{R}$, then $\sup_{x\in\mathbb{R}}V_{\nu}(x,\sigma)$ is continuous on $[0,\infty)$.
\label{corollary: continuity on [0,infty)}
\end{corollary}

\begin{lemma}
$\lim_{\sigma\to\infty} \sup_{x\in\mathbb{R}} V_{\nu}(x,\sigma) = 0$.
\label{lemma: tend to 0 as sigma tend to infty}
\end{lemma}
\begin{proof}
The technique of this proof is similar to that of the previous lemmas. It suffices for us to show that for any sequence $\{\sigma_n\}_{n=1}^{\infty}$ satisfying $\lim_{n\to\infty}\sigma_n=\infty$ (WLOG assume $\{\sigma_n\}_{n=1}^{\infty}\subset(0,\infty)$), the sequence of function $\{V_{\nu}(x,\sigma_n)\}_{n=1}^{\infty}$ tends uniformly to $0$.

So choose a sequence $\{\sigma_n\}_{n=1}^{\infty}$ satisfying the aforementioned properties.

We first observe a useful property of $V_{\nu}(x,\sigma_n)$ for all $x\in\mathbb{R}$ and $n$:
\begin{equation}
\begin{aligned}
    V_{\nu}(x,\sigma_n) 
    = & \int_{\mathbb{R}} \frac{1}{\nu^2} K^2\left(\frac{x-\theta}{\nu}\right)\frac{1}{\sqrt{2\pi}\sigma_n} e^{-\frac{\theta^2}{2\sigma_n^2}}d\theta - \left(\int_{\mathbb{R}} \frac{1}{\nu} K\left(\frac{x-\theta}{\nu}\right)\frac{1}{\sqrt{2\pi}\sigma_n} e^{-\frac{\theta^2}{2\sigma_n^2}}d\theta\right)^2 \\
    \le & \int_{\mathbb{R}} \frac{1}{\nu^2} K^2\left(\frac{x-\theta}{\nu}\right)\frac{1}{\sqrt{2\pi}\sigma_n} e^{-\frac{\theta^2}{2\sigma_n^2}}d\theta \\
    \le & \frac{1}{\nu^2 \sqrt{2\pi}\sigma_n} M \int_{\mathbb{R}} K\left(\frac{x-\theta}{\nu}\right)d\theta \\
    \le & \frac{1}{\sigma_n} C
\end{aligned}
\end{equation}
for some finite constant $C$ independent of $x$ and $n$. This proves uniform convergence.
\end{proof}

\section{Turbulence Simulator}
\subsection{Overview}
When light passes through the atmosphere, the random fluctuations in density cause the light to bend, just as it bends through any medium. The atmosphere acts as a jagged, tilted lens causing the wavefront to bend and become distorted between the source and destination. These effects not only changes over time, as the atmosphere changes and shifts, but also spatially. If we are taking a picture of a wide area at a long distance, the path through the atmosphere between our camera and the source on one side of the image may be very different from the path on another side. If we want to create data in order to test a turbulence mitigation algorithm we must consider the physics at hand, a simple shift-and-blur model will not be an accurate representation.

In order to properly simulate the effects of atmospheric turbulence on an image, a simulator for light propagation through a turbulent medium is required. While simple models for the effects of a propagation through the atmosphere can be used to achieve the correct long exposure effects, there may be a great mismatch with the frame-to-frame PSFs which are difficult to generate without such a simulation. Our implementation is based on approaches proposed in \cite{HardieSimulator, RoggemannSimulator}. A visual overview of our simulation can be seen in Fig. \ref{fig:sim_overview}. Of course, the effects of the atmosphere are incurred on a continuum, however for the simulation this is not feasible. For the purpose of simulating an image through atmospheric turbulence we represent the effects of the atmosphere as a finite sequence of phase screens, shown in Fig. \ref{fig:sim_overview} as $\phi_i$. We then consider the rest to be free space, the "blank" space between the phase screens. This simplification drastically reduces simulation complexity. 

For the types of problems we are interested in with this work, anisoplanatic (spatially variant) effects are desired. This situation can occur when a wide field of view (FOV) of a scene at a long distance is desired. Therefore, we must find a PSF for every pixel in the image in order to properly generate these anisoplanatic effects. For reduction of complexity, we consider a subset of pixels, within the isoplanatic (spatially invariant) angle, and interpolate. This approximation is motivated by the fact that two neighboring pixels in the object plane pass through very similar portions of the atmosphere on their way to the imaging plane. In addition, we also want to create large phase screens that are cropped for each pixel's propagation. This will more closely match the physical situation and generate spatially correlated PSFs within a small neighborhood.

As an overview of our simulation, a subset of pixels, $\mathcal{P}$, in the object plane is selected. Next, we iterate through $\mathcal{P}$, generating a point source centered at some position $(i,j) \in \mathcal{P}$. The point source is then propagated through free space until it reaches the first cropped phase screen, $\psi_1$, which models the phase effects that would've been incurred had the point source not traveled through a vacuum. This process is repeated back and forth, alternating between free space and phase screens, and is referred to as the split-step propagation method \cite{SchmidtTurbBook}.  A function modeling an aperture (e.g. a mask) is then applied to this received wavefront and the point spread function (PSF) can then be obtained through the inverse Fourier Transform. This process is repeated for all $(i,j) \in \mathcal{P}$ due to the spatial varying nature of turbulence. After interpolation, this will then provide us with a set of spatially varying PSFs for every $(i,j) \in X$. A PSF for the pixel $(i,j)$ we will refer to as $h_{(i,j)}(x,y)$. These PSFs are then applied to the object and a simulated propagation of the desired image can then be obtained for a single frame. We may also want to generate a sequence of phase screens, correlated or uncorrelated to match a desired situation, to simulate taking multiple frames through a turbulent medium.

\begin{figure}
    \centering
    \begin{subfigure}{0.4\textwidth}
    \includegraphics[width = 1\linewidth]{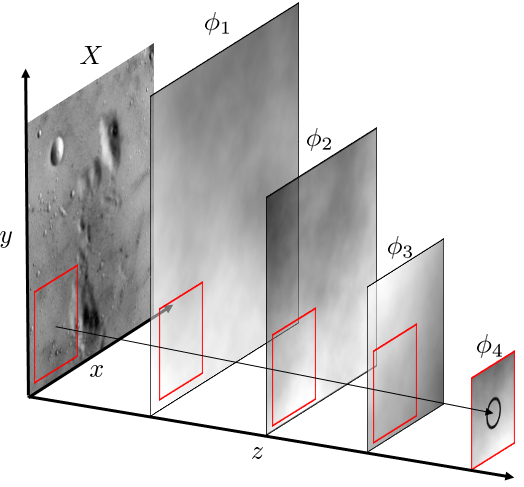}
    \caption{An image, $X$, is simulated as its light propagating through a turbulent medium by the application of free-space propagation and phase screens, $\phi_i$. Point sources, located at a subset of pixels in the imaging plane, are propagated through a cropped section of the larger phase screen (shown in red). This process is repeated across the image plane to obtain accurate spatially variant effects.}
    \label{fig:sim_setup}
    \end{subfigure}
    \hspace{0.03\linewidth}
    \begin{subfigure}{0.45\textwidth}
    \includegraphics[width = 1\linewidth]{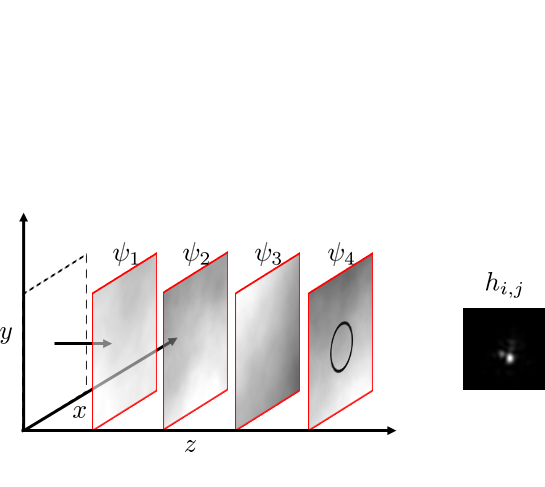}
    \caption{For each propagation required for simulation, a cropped phase screen, $\psi_i$, is used to model the atmosphere for a section of the propagation distance, approximately $L$, and treated as horizontal propagation. The received wavefront is then passed through an aperture (shown as a black circle) and a PSF, $h_{i,j}$, can be generated.}
    \label{fig:sim_setup_horiz}
    \end{subfigure}
    \caption{Overview of Simulation Visualization}
    \label{fig:sim_overview}
\end{figure}

\subsection{Phase Screen Generation}
Phase screens model the phase behavior incurred by atmospheric turbulence over a propagation distance. In reality, these effects are incurred over a continuous $C_n^2(z)$, with $z$ along the direction of propagation and $C_n^2$ as the index of refraction variation. However, this is not practical for the purposes of simulation. Therefore the split-step propagation method is adopted in order to make the simulation feasible. The split-step propagation refers to the two major steps: free space propagation and phase adjustments.

\begin{figure}
\centering
\begin{subfigure}{.32\textwidth}
  \centering
  \includegraphics[width=0.5\linewidth]{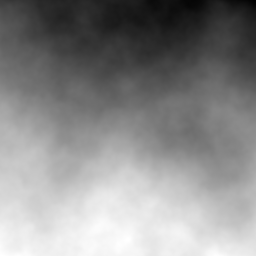}
\end{subfigure}%
\begin{subfigure}{.32\textwidth}
  \centering
  \includegraphics[width=0.5\linewidth]{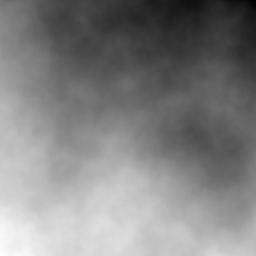}
\end{subfigure}
\begin{subfigure}{.32\textwidth}
  \centering
  \includegraphics[width=0.5\linewidth]{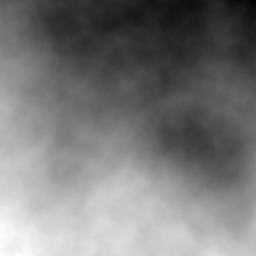}
\end{subfigure}
\par\bigskip
\begin{subfigure}{.32\textwidth}
  \centering
  \includegraphics[width=0.5\linewidth]{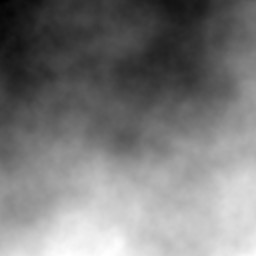}
\end{subfigure}%
\begin{subfigure}{.32\textwidth}
  \centering
  \includegraphics[width=0.5\linewidth]{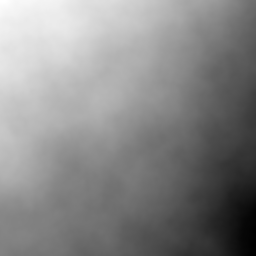}
\end{subfigure}
\begin{subfigure}{.32\textwidth}
  \centering
  \includegraphics[width=0.5\linewidth]{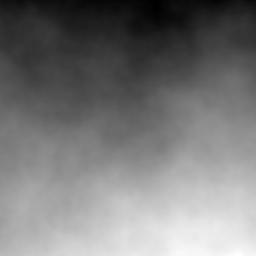}
\end{subfigure}
\caption{Depending on the type of situation required to be simulated, correlated phase screens [Top Row] or uncorrelated phase screens [Bottom Row] may be generated. For example, if a video with a low frame rate is desired, uncorrelated would be preferred.}
\label{fig:phase_screens}
\end{figure}

In order to generate these phase screens, there is no direct distribution we can pull from. Instead the power spectral density (PSD) of phase screens, a well studied phenomenon, is used. Of the models that exist, the (Modified) Von K\'arm\'an and Kolmogorov PSDs are by far the most common. For the generation of the phase screen, a general approach is as follows:
\begin{enumerate}
    \item Generate phase screen Power Spectral Density (PSD) based on system parameters. For the purposes of this work the modified Von K\'arm\'an PSD was chosen and is given by
\begin{equation*}
    S(\rho) = \frac{0.023e^{-\rho/(\frac{5.92}{2\pi} l_0)^2}}{r_0^{5/3}(\rho^2 + (1/L_0)^2)^{11/6}}
\end{equation*}
where $r_0$ is the Fried parameter which is, roughly speaking, the diameter of a circle we can draw on the received wavefront and still consider it to be relatively flat (within a radian). Furthermore, $p$ is $\sqrt{u^2 + v^2}$ with $u, v$ as the frequency components, and $l_0, L_0$ as the small and large eddy size, respectively. It should be noted that the low frequency components usually require additional effort, as they must be generated extremely accurately. These subharmonic methods are described in great detail in \cite{SchmidtTurbBook}.
    \item If we consider $A, B$ to be 2D white Gaussian processes, we can compute
    \begin{equation*}
        C= (A + jB) \cdot \sqrt{S(\rho)}
    \end{equation*}
    with $\cdot$ representing element-wise multiplication and $C$ being the PSD imparted with noise. This process can be replaced by a filtering operation, filtering Gaussian white noise with the PSD. If generating sequence of correlated phase screens is desired, this is possible through correlation of the noise.
    \item In order to generate the phase screen, we then compute the real component of the inverse Fourier Transform, 
    \begin{equation*}
        \phi_i(x,y) = \text{Real}\left( \mathcal{F}^{-1}\{ C \} \right)
    \end{equation*}
    and we receive the phase screen, $\phi_i$. In our simulation, these are made bigger than required and then cropped as previously mentioned.
\end{enumerate}

\subsection{Propagation over Optical Path}
In this section, we will discuss the four major components of the propagation over the optical path: the selected point source, free space propagation, using the phase screens, and PSF generation. The steps that are described take place over all locations $(i,j) \in \mathcal{P}$, and for any number of frames as desired. This discussion is an elaboration of the process shown in Fig. \ref{fig:sim_setup_horiz}. A point source is convolved with a transfer function representing free space, and the phase screens are then used to impart the atmospheric effects. After this, an aperture mask is applied and the PSF can be obtained through the inverse Fourier Transform.

\subsubsection{Point Sources Across Image}
In this simulation we are interested in simulating wide field of view scenarios. In order to obtain the associated effects, we must propagate point sources through the phase screens beginning at different locations across the image. To this end, we select a subset of points in the object plane, $\mathcal{P}$. Point sources centered on this subset are then propagated through the phase screens and free-space. This gives us a spatial correlation in our PSFs, since each point that is propagated through will be traveling through a phase screen sharing some elements with its neighboring propagation.

For the source, a Gaussian windowed sinc squared function is chosen. This source selection is to allow for easily obtaining a flat region over the aperture disc. This point source is of the form

\begin{equation*}
    u_0(x,y) = \lambda L R^2e^{-\frac{jk}{2L}(x^2 + y^2)}\text{sinc}(Rx)\text{sinc}(Ry)e^{-\frac{R^2}{16}(x^2 + y^2)}
\end{equation*}
where $R = F/(\lambda L)$ and $F$ is the size of the flat region desired in the imaging plane. $F$ is chosen to be between the size of the aperture and propagation phase screen size. The motivation for doing this can be seen in Fig. \ref{fig:point_source}. Using a simple impulse function, we may not get equal activation across the aperture through free space, however this is achieved quite easily through use of this particular point source and Fresnel diffraction.

\begin{figure}
\centering
\begin{subfigure}{.32\textwidth}
  \centering
  \includegraphics[width=1\linewidth]{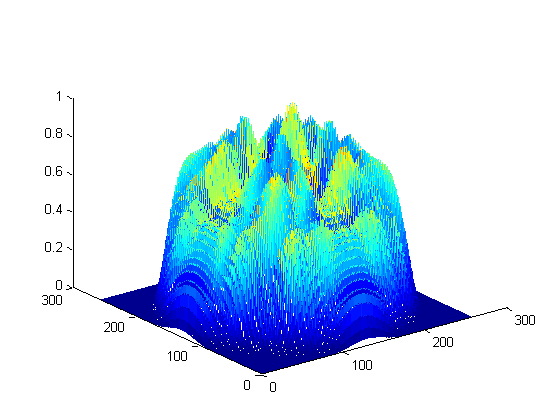}
  \caption{Example of a free space propagation of an impulse function. This result shows that this selection of a source may not be ideal, since the aperture may not receive a flat wavefront}
  \label{fig:source_delta}
\end{subfigure}%
\hspace{0.1\textwidth}
\begin{subfigure}{.32\textwidth}
  \centering
  \includegraphics[width=1\linewidth]{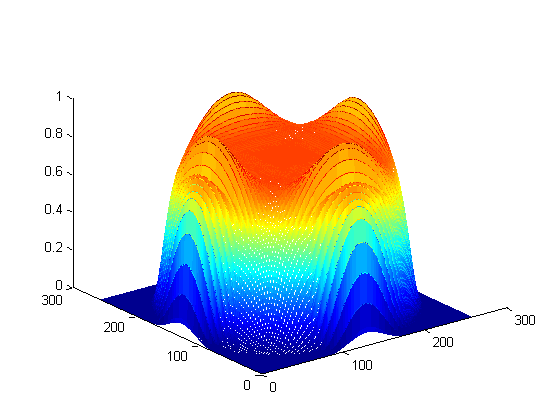}
  \caption{Example of a free space propagation with the chosen point source function. This selection allows there to be a flat region across the aperture that is easily tunable.}
  \label{fig:no_atm_source}
\end{subfigure}
\caption{Results of free space propagation with different sources}
\label{fig:point_source}
\end{figure}

\subsubsection{Free Space Propagation}
After a point source has been generated, it must propagate through free space to reach the first phase screen. There are two major ways of modeling diffraction, namely Fraunhofer and Fresnel. For this work, our propagations take place across what may be viewed as a long distance (e.g. 6 km), though they do not fall into what is considered to be far field. Therefore, Fresnel diffraction is chosen as our model for free space propagation.  This selection is also convenient, since it can be represented as a convolution through
\begin{equation*}
    h(x,y,z) = \frac{e^{jkz}}{j\lambda z} e^{j\frac{k}{2z}(x^2+y^2)}
\end{equation*}
with $x, y$ spatially and $z$ along the axis of propagation. In addition, $k$ is the wavenumber, $\frac{2\pi}{\lambda}$. We can then perform a 2D convolution with $h$ and our wavefront, $u(x,y)$, from the source to the first phase screen, or the $i^{th}$ phase screen to the $(i+1)^{th}$ phase screen. In our simulation, the distance between phase screens is fixed, leaving $z$ as a constant.

\begin{figure}[h]
\centering
\begin{tabular}{ccc}
\includegraphics[width=0.3\linewidth]{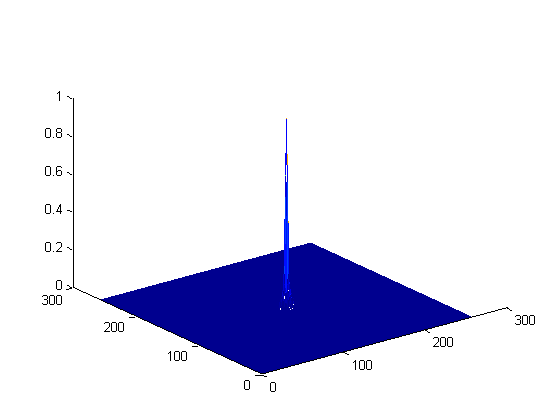}&
\includegraphics[width=0.3\linewidth]{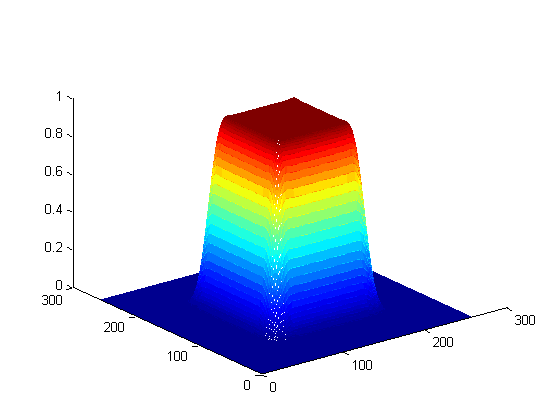}&
\includegraphics[width=0.3\linewidth]{supp_pix/stage_final_no_atm.png}\\
(a) & (b) & (c)
\end{tabular}
\caption{Visualization of propagation through \textbf{free space}. (a) The chosen source, a Gaussian windowed sinc squared function. (b) A view of the wavefront halfway to the imaging plane through free space. Notice the flat wavefront which will easily allow the effects of the phase screen to be properly imparted. (c) A view of the wavefront as seen at the aperture. The flat region was selected to be considerably larger than the aperture, though smaller than the width of the propagation phase screen size.}
\label{fig:free_space}
\end{figure}

A crucial component of this simulation is the sampling constraints. For our purposes, we have use the Voelz sampling criteria to solve the issues that may arise during propagation. The sampling constraint is given by

\begin{equation*}
    \Delta x = \sqrt{\frac{\lambda L}{N}}
\end{equation*}
with $N$ as the cropped phase screen size in pixels, and $L$ as the propagation path length.

\subsubsection{Imparting of Phase}
After the wavefront has propagated to reach phase screen $\phi_i$, and referring to this wavefront as $u_i(x,y)$, we can then impart the phase on the wavefront simply by performing
\begin{equation*}
    \tilde{u}_i(x,y) = u_i(x,y)e^{j\psi_i(x,y)} = [h(x,y,z) \circledast  \tilde{u}_{i-1}(x,y)]e^{j\psi_i(x,y)}
\end{equation*}
where $\tilde{u}$ is the wavefront after being imparted with phase and $\psi_i$ is the $i^{th}$ cropped phase screen. For the purpose of speeding this process up, these convolutions are carried out as elementwise products in the frequency domain. A way of visualizing this is shown in Fig. \ref{fig:sim_setup_horiz}, where we can think of this process as a horizontal propagation using the same cropping as in Fig. \ref{fig:sim_setup}. Fig. \ref{fig:atm_prop} shows the effects the phase screens have on the wavefront as compared with Fig. \ref{fig:free_space}, where there is no atmosphere present. It also should be noted that a windowing function is used during these propagations, in order to avoid a great deal of the wavefront reaching the edges of the propagation window.

\begin{figure}[h]
\centering
\begin{tabular}{ccc}
\includegraphics[width=0.3\linewidth]{supp_pix/source_gaussian.png}&
 \includegraphics[width=0.3\linewidth]{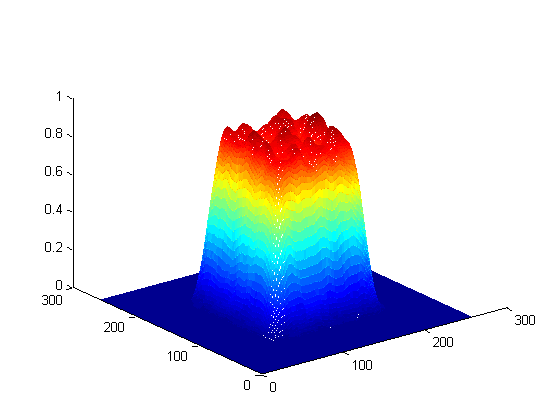}&
 \includegraphics[width=0.3\linewidth]{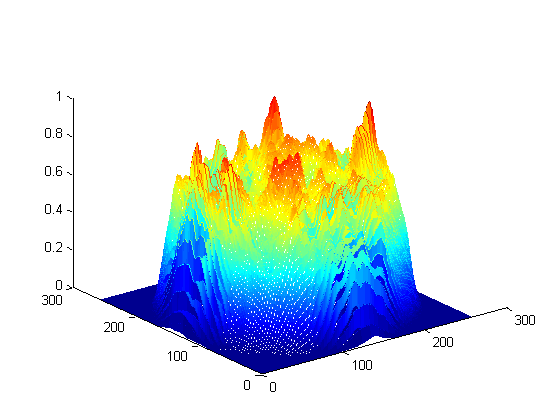}\\
 (a) & (b) & (c)\\
 \end{tabular}
 \caption{Propagation through \textbf{atmosphere}. (a) The chosen source, a Gaussian windowed sinc squared function. (b) A view of the wavefront halfway to the imaging plane through the modeled atmosphere. Notice the similarity to the free-space wavefront, although with some "dents" in the wavefront. Here the effects of the phase screen can be seen clearly. (c) A view of the wavefront as seen at the aperture after propagation through the atmosphere. It is nearly unrecognizable as compared to its free space counterpart. The effects of the phase screens early on have accumulated and caused a considerable amount of distortion.}
 \label{fig:atm_prop}
\end{figure}

\subsubsection{Aperture Model and PSF Generation}
After the final free space propagation has taken place and the wavefront is now at the imaging plane, we can consider the effects of the aperture, shown as a black circle in Fig. \ref{fig:sim_setup} and Fig. \ref{fig:sim_setup_horiz}. For this we write the received wavefront as

\begin{equation*}
    r(x,y) = a(x,y)\tilde{u}_N(x,y)\exp\left(\frac{-j\pi(x^2+y^2)}{\lambda L}\right)
\end{equation*}
where $a(x,y)$ is a function representing the aperture (such as a ciruclar/square mask), and the final term is a lens focusing operation. We may then find the PSF of the pixel at $(i,j)$ to be

\begin{equation*}
    h_{(i,j)}(x,y) = [\mathcal{F}\{r(x,y)\}]_{u=\frac{x}{\lambda L}, v=\frac{y}{\lambda L}}.
\end{equation*}
The results of this PSF generation can be seen in Fig. \ref{fig:aperture} with two different types of apertures. After all PSFs have been generated on the subset of pixels, they are then interpolated to provide a PSF for each pixel location. It is also important to note that the PSFs are normalized to have a sum of 1.

\begin{figure}
\centering
\begin{subfigure}{.25\textwidth}
  \centering
  \includegraphics[width=0.5\linewidth]{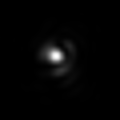}
  \caption{Example PSF using a circular aperture.}
  \label{fig:circ_ap}
\end{subfigure}%
\hspace{0.1\textwidth}
\begin{subfigure}{.25\textwidth}
  \centering
  \includegraphics[width=0.5\linewidth]{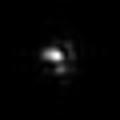}
  \caption{Example PSF using square aperture.}
  \label{fig:square_ap}
\end{subfigure}
\caption{PSFs with Different Apertures}
\label{fig:aperture}
\end{figure}

\subsection{Obtaining Simulated Image from PSFs}
After propagating each point source through the optical path and obtaining the PSF at the imaging plane, we must then perform a spatially variant convolution. This gives us the anisoplanatic effects that are a crucial feature in a wide FOV imaging scenario through the atmosphere. After this is performed, the simulated output image is obtained. We may then repeat the process for as many frames as desired, with correlated or uncorrelated phase screens. A pseudocode for single frame generation is shown below.

 \begin{algorithm}
 \caption{Algorithm for Simulation of Single Image Through Turbulence}
 \begin{algorithmic}[1]
 \label{alg:simulation_alg}
 \renewcommand{\algorithmicrequire}{\textbf{Input:}}
 \renewcommand{\algorithmicensure}{\textbf{Output:}}
 \REQUIRE image $X$, geometric/turbulence parameters
 \ENSURE $\tilde{X}$
  \STATE generate the phase screens, $\phi_l$ for $l \in [0,N-1]$
  \FOR {$(i,j) \in \mathcal{P}$}
  \FOR {$k \in [0,N-1]$}
   \STATE select cropped phase screen from larger phase screen, $\psi_k(x,y)$\
   \STATE $\tilde{u}_{k+1} = [h(x,y,z) \circledast  \tilde{u}_{k}(x,y)]e^{j\psi_k(x,y)}$
   \STATE $r(x,y) = a(x,y)\tilde{u}_{N-1}(x,y)\exp\left(\frac{-j\pi(x^2+y^2)}{\lambda L}\right)$
   \STATE $h_{(i,j)}(x,y) = [\mathcal{F}\{r(x,y)\}]$\;
   \STATE resample $h_{(i,j)}$ to imaging plane, with size $(M+1)\times (M+1)$
   \ENDFOR
  \ENDFOR
  \STATE interpolate $h_{(i,j)}$ for all $(i,j) \in X$
  \FOR {$(i,j) \in \mathcal{P}$}
  \STATE $\tilde{X}(i,j) = \sum\limits_{x=-M/2}^{M/2}\sum\limits_{y=-M/2}^{M/2} X(i-x,j-y)h_{i,j}(x,y)$
  \ENDFOR
 \RETURN $\tilde{X}$ 
 \end{algorithmic} 
 \end{algorithm}
 
\newpage
This work has been cleared for public release carrying the approval number 88ABW-2019-2438.

{\small
\bibliographystyle{ieee}
\bibliography{egbib}
}